\documentclass[11pt,a4paper]{article}
\usepackage{threeparttable}
\usepackage{stmaryrd}
\usepackage{booktabs}
\usepackage{amssymb,cite}
\usepackage{mathrsfs}
\usepackage{latexsym}
\usepackage{epsfig}
\usepackage{tikz}
\usepackage{subfigure}
\usepackage{graphicx}
\usepackage{epstopdf}
\usepackage{amsthm}
\usepackage{amsfonts}
\usepackage{amsmath}
\usepackage{amssymb,amsmath,amsthm,amscd}
\usepackage{algorithm}
\usepackage{algorithmic}
\usepackage{multirow}
\usepackage{xcolor}
\usepackage{array}
\usepackage{booktabs}
\usepackage{bm}
\usepackage[numbers,sort&compress]{natbib}
\usepackage{lscape}
\usepackage[indentafter]{titlesec}
\usepackage[colorlinks,linkcolor=red,linkcolor=blue,linktocpage]{hyperref}
\usepackage{threeparttable}
\usepackage{caption}
\usepackage{enumitem,calc}
\newcommand\preitem{\mdseries\textbullet\space}
\newlist{desclist}{description}{3}
\setlist[desclist,1]{format=\preitem\bfseries,leftmargin=\widthof{\preitem},style=sameline}
 \makeatletter

\newcommand{\Rmnum}[1]{\expandafter\@slowromancap\romannumeral #1@}

\usepackage[top=1.0 in,bottom=1.0 in,left=1.0 in,right=1.0 in]{geometry}
\floatname{algorithm}{Method}
\newtheorem{theorem}{Theorem}[section]
\newtheorem{corollary}{Corollary}[section]

\newtheorem{lemma}{Lemma}[section]
\newtheorem{remark}{Remark}[section]

\newtheorem{proposition}{Proposition}[section]
\hypersetup{colorlinks=true, citecolor=blue}

\def \b#1{\overline{#1}}
\def \u#1{\underline{#1}}
\def \W#1{\widehat{#1}}
\def \t#1{\widetilde{#1}}

\numberwithin{equation}{section}

\title{Nonlocal massive Thirring model and its solutions}

\author{Cong-han Wang$^{1}$, ~~ Shu-zhi Liu$^2$,~~
Jing Wang$^{3}$,~~ Da-jun Zhang$^{4,5}$\footnote{
Corresponding author. Email: djzhang@staff.shu.edu.cn}\\
{\small $^{1}$School of Mathematics and Statistics, Zhengzhou Normal University, Zhengzhou 450044, China}\\
{\small $^{2}$School of Statistics and Data Science, Ningbo University of Technology, Ningbo, 315211, China}\\
{\small $^{3}$School of Mathematics, Physics and Statistics, Shanghai Polytechnic University, Shanghai 201209, China}\\
{\small $^{4}$Department  of Mathematics, Shanghai University, Shanghai 200444, China}\\
{\small $^{5}$Newtouch Center for Mathematics of Shanghai University,  Shanghai 200444,  China}
}

\date{\today}

\begin{document}

\maketitle

\begin{abstract}
A nonlocal version of the massive Thirring model (MTM) and its solutions are presented.
We start from a 4-component system that can be reduced to the classical MTM and nonlocal MTM.
Bilinear form of the 4-component system and general double Wronskian solutions are derived.
By utilizing reduction technique we obtain solutions
of the nonlocal MTM. Relations between the nonlocal MTM and the nonlocal Fokas-Lenells
equation is discussed.
Some solutions of the nonlocal MTM, such as solitons, double-pole solution,
algebraic solitons and high order algebraic solitons are analyzed and illustrated.


\end{abstract}


\section{Introduction}

The massive Thrring model (MTM) bears the name of Walter E. Thirring \cite{T-1958}
as he provided it as a solvable model in the relativistic Quantum Field Theory (QFT).
Thirring solved the massless case \cite{T-1958},
but at that time the massive case was difficult (see \cite{W-1967}).
It is Mikhailov  \cite{M-1976} who first revealed the integrability of the two-component MTM
\begin{equation}\label{MTM}
(-i\partial_{\mu}\gamma^{\mu}+m)\psi
+g\gamma^{\mu}\psi (\overline{\psi}\gamma_{\mu}\psi)=0.
\end{equation}
Around the same time, this model was also connected to the quantum and classical sine-Gordon equations
\cite{C-1975,O-1976}.
Here in the MTM model \eqref{MTM}, $m$ stands for mass, $g$ is a parameter,
$i$ is the imaginary unit,
$\psi=(w,z)^T$ is a two-dimensional spinor field,
$\gamma_0=\bigl(\begin{smallmatrix}0 & 1 \\ 1 & 0\end{smallmatrix}\bigr)$,
$\gamma_1=\bigl(\begin{smallmatrix}0 & -1 \\ 1 & 0\end{smallmatrix}\bigr)$,
$\overline{\psi}=(w^*,z^*)\gamma^{0}$ where $*$ denotes complex conjugate,
$\partial_\mu=\partial_{x_\mu}$ for $\mu=1,2$,
$\gamma^{\mu}=(\gamma_{\mu})^{-1}$ and the Einstein summation convention is employed.
The model is used to describe interaction of two states of a fermion.
In terms of light-cone coordinates, the MTM \eqref{MTM} is written as (with $m=2, g=1$)
\begin{subequations} \label{cMTM}
\begin{align}
&  iw_{x}+z+w|z|^2=0,\label{KNa}\\
&  iz_{t}+w+z|w|^2=0,\label{KNb}
\end{align}
\end{subequations}
which can be resolved by the following equation \cite{KN-1977,GIK-1980}
\begin{equation}\label{FL3}
q_t+u-2i|u|^2 q=0,~~ u_x=q
\end{equation}
via
\begin{equation}\label{tran-M-MTM}
z = q^* e^{-i\beta},~~ w=-i u^*e^{-i\beta},~~
 \beta=\int^{x}_{-\infty}|q|^2\mathrm{d}x,
 ~~q=u_x.
\end{equation}
The later equation \eqref{FL3} in terms of $u$, i.e.
\begin{equation}\label{FL}
u_{xt}+u-2i|u|^2 u_x=0,
\end{equation}
is a reduction of the so-called Mikhailov model (see \cite{GIK-1980})
\begin{subequations} \label{pKN-1}
\begin{align}
&  u_{xt}+u-2iuvu_x=0,\label{KNa}\\
&  v_{xt}+v+2iuvv_x=0,\label{KNb}
\end{align}
\end{subequations}
which is the first member in the negative (potential) Kaup-Newell (KN) hierarchy.
Details of the above relations were also collected in Appendix A of \cite{LWZ-2022}.
As a solvable QFT model and also one of early integrable models,
the MTM has received a lot of attention.
One can refer to \cite{KM-1977,IK-1978,D-1978,D-1979,BT-1979,KMI-1979,PG-1979,P-1981,NCQV-1983,NCQ-1983,
DHS-1984,B-1985,BG-1987,V-1991}
for the  investigations in the early days of the this equation.
More recent results about the model can be found from
\cite{GWCH-CNSNS-2017,Pelinovsky-2019,Xia-PRSE-2020,LWZ-2022,CF-2022,CYF-SAPM-2023,
HHP-PRE-2024,HLQ-2025}.

It is interesting that  equation \eqref{FL} was also derived (with some transformations, see \cite{Lenells-SAPM-2009})
as a novel generalization of the nonlinear Schr\"odinger  equation
using bi-Hamiltonian structures of the later \cite{Fokas-PD-1995,Lenells-F-Non-2009},
as well as derived from modeling  propagation of nonlinear  pulses in monomode optical fibers \cite{Lenells-SAPM-2009}.
Therefore equation \eqref{FL} is also known as the Fokas-Lenells (FL) equation.
This equation has also been well studied from many aspects, such as dressing method \cite{Lenells-JNS-2010},
Darboux transformation \cite{Wang-XL-AML-2020},
bilinear approach\cite{LWZ-2022,Vek-Non-2011,Matsuno-JPA-2012a,Matsuno-JPA-2012b},
Riemann-Hilbert problem\cite{CF-JDE-2022}, algebraic geometry solutions \cite{Zhao-F-JNMP-2013},
rogue waves \cite{Xu-HCP-MMAS-2014}, algebraic solitons \cite{W-2021,WW-2022},
nonlocal form\cite{LWZ-2022}, and so on.
Nonlocal integrable systems were  first systematically introduced by Ablowitz and Musslimani in 2013
\cite{AblM-PRL-2013} and have received intensive attention since then
(e.g.\cite{ChenDLZ-SAPM-2018,GP-JMP-2018,Zhou-SAPM-2018,AblFLM-SAPM,YY-SAPM-2018,AblM-JPA-2019,
Lou-SAPM-2019,BioW-SAPM-2019,AblLM-Nonl-2020,RaoCPMH-PD-2020,
LiFW-SAPM-2020,RybS-CMP-2021,AM-PLA-2021}).
The FL equation \eqref{FL} has a nonlocal version
\begin{equation}\label{non-FL}
u_{xt}(x,t)+u(x,t)-2iu(x,t)u(-x,-t)u_x(x,t)=0,
\end{equation}
which is reduced from the Mikhailov model \eqref{pKN-1} by taking $v(x,t)=u(-x,-t)$.
Considering the close relation between the MTM and FL equation,
there should be a  nonlocal counterpart of the MTM \eqref{MTM}.

This paper is devoted to exploring nonlocal MTM.
We will consider a four-component system \cite{BG-1987}
\begin{subequations} \label{31pKN-1}
\begin{align}
&  iq_{1,t}+q_{3}+q_{1}q_{3}q_{4}=0,\label{KN11}\\
&  iq_{2,t}-q_{4}-q_{2}q_{3}q_{4}=0,\label{KN21}\\
&  iq_{3,x}+q_{1}+q_{1}q_{2}q_{3}=0,\label{KN31}\\
&  iq_{4,x}-q_{2}-q_{1}q_{2}q_{4}=0.\label{KN41}
\end{align}
\end{subequations}
It allows a nonlocal reduction, which yields a nonlocal MTM
\begin{subequations} \label{non-cMTM1}
\begin{align}
&  iw_{t}(x,t)+z(x,t)+w(x,t)z(x,t)z(-x,-t)=0,\label{KNa}\\
&  iz_{x}(x,t)+w(x,t)+z(x,t)w(x,t)w(-x,-t)=0.\label{KNb}
\end{align}
\end{subequations}
Solutions of this nonlocal system will be derived by means of the bilinearisation-reduction approach.
This approach is particularly developed for solving nonlocal systems
\cite{ChenDLZ-SAPM-2018,ChenZ-AML-2018,Deng-AMC-2018}.
It is based on bilinear equations and solutions in double Wronskian/Casoratian form.
For those integrable equations that are reduced from coupled systems,
one can first solve the unreduced coupled system to obtain their solutions in double Wronskian/Casoratian form,
then, implement suitable reduction techniques
so that solutions of the reduced equation can be obtained as reductions of those of the unreduced coupled system.
The method has proved effective in solving not only nonlocal integrable equations
but also the classical (local) ones \cite{ChenDLZ-SAPM-2018,ChenZ-AML-2018,Deng-AMC-2018,
ShiY-ND-2019,Liu-ROMP-2020,ZSL-2020,LWZ-ROMP-2022,WW-CNSNS-2022,ZLD-2023,DCCZ-2024}.
One can refer to the recent review \cite{ZDJ-review-2023} for this method.
In addition, considering the relation between the MTM and  FL equation,
we will also explain relations between the nonlocal MTM and nonlocal FL equation.

The paper is organized as follows.
In Sec.\ref{sec-2} we present integrability of the nonlocal MTM   \eqref{non-cMTM1}
and explore the connection with the nonlocal FL equation.
Sec.\ref{sec-3} derives solutions using the bilinearisation-reduction approach.
Dynamics of some obtained solutions are analyzed in Sec.\ref{sec-4}.
Finally, conclusions are presented in Sec.\ref{sec-5}.
There are two appendices.
Appendix \ref{app-A} consists of a proof of Theorem \ref{Theo-3-1},
and Appendix \ref{app-B} discusses an alternative bilinear formulation
of \eqref{31pKN-1}.

\section{Nonlocal massive Thirring model}\label{sec-2}

In this section we present integrability of the nonlocal MTM \eqref{non-cMTM}
and also explore the connection with the nonlocal FL equation.

We start from the four-component system \eqref{31pKN-1},
which is integrable with the following Lax pair (cf.\cite{BG-1987,BGK-1993})
\begin{equation}
\Phi_{x}=M\Phi, ~~
~~\Phi_{t}=N\Phi,
\end{equation}
where $\Phi=(\phi_1, \phi_2)^T$,
\begin{subequations} \label{4pKN-1}
\begin{align}
&  M=\frac{i}{2}\begin{pmatrix}
-\lambda^{2}+q_{1}q_{2} & 2\lambda q_{1}\\
2\lambda q_{2} & \lambda^{2}-q_{1}q_{2}
\end{pmatrix},\label{La}\\
&  N=\frac{i}{2}\begin{pmatrix}
-\lambda^{-2}+q_{3}q_{4} & 2\lambda^{-1}q_{3}\\
2\lambda^{-1}q_{4} & \lambda^{-2}-q_{3}q_{4}
\end{pmatrix},\label{Lb}
\end{align}
\end{subequations}
and $\lambda$ is the spectral parameter.
The compatibility of $\Phi_{xt}=\Phi_{tx}$ gives rise to the zero curvature  equation $ M_t-N_x+[M,N]=0$
and the system \eqref{31pKN-1} follows.
Note that the system \eqref{31pKN-1} allows a conserved relation
\begin{equation}\label{con-q}
(q_1 q_2)_t=-(q_3 q_4)_x,
\end{equation}
which will be used later.
Let us call \eqref{31pKN-1} the unreduced MTM.
In fact, by imposing reduction
\begin{equation}
\label{trans a}
 q_1=\delta  {q_2}^*=z,\\
~~ q_3=\delta {q_4}^*=w, ~~~(\delta =\pm 1),
\end{equation}
one immediately gets the MTM \eqref{MTM}, i.e.
\begin{subequations} \label{MTM-q}
\begin{align}
&  iq_{1,t}+q_{3}+ {\delta} q_{1}|q_{3}|^2=0,\label{KNa-q}\\
&  iq_{3,x}+q_{1}+ {\delta} q_{3}|q_{1}|^2=0.\label{KNb-q}
\end{align}
\end{subequations}
If we consider a nonlocal
(reversed space and time) reduction
\begin{equation}
\label{trans b}
q_1(-x,-t)=\delta q_2(x,t), ~~ q_3(-x,-t)=\delta q_4(x,t),
\end{equation}
the four-component system \eqref{31pKN-1} reduces to
\begin{subequations} \label{non-cMTM}
\begin{align}
&  iq_{1,t}+q_{3}+\delta q_{1}q_{3}q_{3}(-x,-t)=0,\label{non-cMTMa}\\
&  iq_{3,x}+q_{1}+\delta q_{3}q_{1}q_{1}(-x,-t)=0,\label{non-cMTMb}
\end{align}
\end{subequations}
i.e. the nonlocal MTM \eqref{non-cMTM1} with $q_1=w, q_3=z$.

Next, we explore the relations between the unreduced MTM \eqref{31pKN-1} and the Mikhailov model
\eqref{pKN-1} and also their reductions.
Consider \eqref{pKN-1} and we rewrite it into the form
\begin{subequations} \label{2pKN-1}
\begin{align}
&  q_{t}+u-2iuvq=0,\label{KNc}\\
&  r_{t}+v+2iuvr=0,\label{KNd}
\end{align}
\end{subequations}
where
\begin{equation}
\label{trans 1}
q=u_{x},~~ r=v_{x}.
\end{equation}
Note that  \eqref{pKN-1} has a conservation law (cf.\eqref{con-q})
\begin{equation}\label{conser}
 (qr)_t=-(uv)_x,
\end{equation}
if $uv \to 0$ as $x\to \pm \infty$.
Introduce transformation
\begin{equation}
\label{trans 3}
q=q_1\, e^{-i\beta},
~~r=q_2\, e^{i\beta},
~~u=-i q_3\,e^{-i\beta},
~~v=i q_4\, e^{i\beta},
\end{equation}
where\footnote{The definition of $\partial^{-1}_x$ indicates $\partial^{-1}_{-x}=-\partial^{-1}_{x}$.}
\begin{equation}\label{beta}
\beta=\partial^{-1}_x qr,~~
\partial^{-1}_x=\frac{1}{2}\left(\int_{-\infty}^{x}-\int^{+\infty}_{x}\right) \cdot \,\mathrm{d}x.
\end{equation}
Substituting the above  $q, u, v$ into \eqref{KNc} yields
\[
  q_{t}+u-2iuvq=
  e^{-i\beta} (q_{1,t}-i \beta_t q_1 -i q_3 -2i q_1 q_3 q_4).
\]
Noticing from \eqref{conser} that $\beta_t=-uv=-q_3 q_4$, we have
\[
  q_{t}+u-2iuvq=
 -i e^{-i\beta} (i q_{1,t}+ q_3 + q_1 q_3 q_4),
\]
which gives rise to equation \eqref{KN11} as a consequence of \eqref{KNc}.
In a similar way, one can get \eqref{KN21} from \eqref{KNd}.
Next, from $i q_3=-u\,e^{i\beta}$ and making use of \eqref{trans 1} and \eqref{trans 3} we have
\[
iq_{3,x} =-u_x e^{i\beta} - i u q r e^{i\beta}
             =-q e^{i\beta} - i u q r e^{i\beta}
            =-q_1 -q_1 q_2 q_2,
\]
i.e. equation \eqref{KN31}. The last equation \eqref{KN41} follows from $iq_4=v e^{i\beta}$.
Thus, we arrive at the following.

\begin{proposition}\label{prop-2-1}
If $u$ and $v$ are solutions of the Mikhailov model \eqref{pKN-1} and $uv$ tends to $0$ as $x \to \pm \infty$,
then $\{q_i\}$ defined below,
\begin{equation}
\label{trans 33}
q_1=q\, e^{i\beta}=u_x \, e^{i\beta},
~~q_2=r\, e^{-i\beta} =v_x\, e^{-i\beta},
~~q_3=i u\,e^{i\beta},
~~q_4=-iv\, e^{-i\beta},
\end{equation}
satisfy the unreduced MTM \eqref{31pKN-1}. Here $\beta$ is defined in \eqref{beta}.
\end{proposition}

It can be examined that the relation \eqref{trans 33} is consistent with the reduction \eqref{trans a}
and $u=v^*$, which lead the unreduced MTM \eqref{31pKN-1} to the MTM \eqref{MTM-q} and lead
the Mikhailov model \eqref{pKN-1} to the FL equation \eqref{FL}, respectively.
Thus, once $u$ solves the FL equation \eqref{FL} and $|u|\to 0$ as $x\to \pm \infty$,
then
\begin{equation}\label{qu}
q_1= u_x \, e^{i\beta},
~~q_3=i u\,e^{i\beta}, ~~ (\beta=\partial^{-1}_x |q|^2)
\end{equation}
satisfy the MTM \eqref{MTM-q}.
The connection \eqref{qu} has been known since 1970s, cf.\cite{GIK-1980,KN-1977}.

Next, we examine the connection between the nonlocal MTM   and the nonlocal
FL equation and the consistency of the involved reductions.
For the Mikhailov model \eqref{pKN-1}, the nonlocal reduction 
\begin{equation}\label{red-n}
v(x,t)=u(-x,-t)  
\end{equation}
yields the nonlocal FL equation
\begin{equation}\label{non-FL-shift}
u_{xt}(x,t)+u(x,t)-2iu(x,t)u(-x,-t)u_x(x,t)=0.
\end{equation}
Again, introducing $q=u_x$ and rewriting it into
\begin{equation}\label{non-FL-shift-q}
q_{t}(x,t)+u(x,t)-2iu(x,t)u(-x,-t)q(x,t)=0,
\end{equation}
it implies that
\begin{equation}
-q_{t}(-x,-t)+u(-x,-t)-2iu(-x,-t)u(x,t)q(-x,-t)=0,
\end{equation}
which then, together with \eqref{non-FL-shift}, gives rise to a conservation law
\begin{equation}\label{con-qq}
(q(x,t) q(-x,-t))_t=(u(x,t)u(-x,-t))_x,
\end{equation}
provided $u(x,t)u(-x,-t) \to 0$ as $x\to \pm \infty$.
To connect the nonlocal MTM \eqref{non-cMTM} and the nonlocal FL equation \eqref{non-FL-shift},
we consider transformations
\begin{equation}\label{qu-non}
q=q_1 \, e^{-i \bar{\beta}},~~~ u=-i q_3 \, e^{-i \bar{\beta}},
\end{equation}
where
\begin{equation}\label{beta-b}
\bar\beta=\partial_x^{-1} q(x,t)q(-x,-t)
\end{equation}
and $\partial_x^{-1}$ is defined as in \eqref{beta}.
Note that with this setting it follows that
$\bar\beta(x,t)=-\bar\beta(-x,-t)$ and
\[ u(-x,-t)=-i q_3(-x,-t) \, e^{-i \bar{\beta}(x,t)}.\]
Then, substituting \eqref{qu-non} into equation \eqref{non-FL-shift-q}
and making use of the conservation law \eqref{con-qq} yield
\begin{align*}
& q_{t}(x,t)+u(x,t)-2iu(x,t)u(-x,-t)q(x,t)\\
=& e^{-i \bar{\beta}(x,t)}\Big[q_{1,t}(x,t)-iq_1(x,t) u(x,t)u(-x,-t)- i q_3(x,t)\\
 & ~~~~~~~~~~~ -2 i q_1(x,t)q_3(x,t) q_3(-x,-t)\Big]\\
=& -i e^{-i \bar{\beta}(x,t)}\Big[iq_{1,t}(x,t)+ q_3(x,t)+ q_1(x,t) q_3(x,t)q_3(-x,-t)\Bigr],
\end{align*}
which gives rise to \eqref{non-cMTMa}. In addition, from $u=-i q_3 \, e^{-i \bar{\beta}}$ we have
\[u_x=e^{-i \bar{\beta}}(-i q_{3,x} - q_3 q(x,t)q(-x,-t)).\]
After replacing $u_x$ with $q$ and replacing $q$ and $u$ using \eqref{qu-non},
we immediately get equation \eqref{non-cMTMb}.

One can check that the relations \eqref{trans 3} and \eqref{qu-non} and the reductions
\eqref{trans b} and \eqref{red-n} are consistent, as in Fig.\ref{F1}.

\begin{figure}[H]
\centering
\begin{tikzpicture}
\draw[<-] (0,2.5) -- (3,2.5);
\node[above] at (1.5,2.5) {  \eqref{trans 3}};
\node[right] at (3,2.5) {\eqref{pKN-1}};
\node[left] at (0,2.5) {\eqref{31pKN-1}};
\draw[->] (-0.5,2) -- (-0.5,0.5);
\draw[->] (3.5,2) -- (3.5,0.5);
\draw[<-] (0,0) -- (3,0);
\node[above] at (1.5,0) {\eqref{qu-non}};
\node[right] at (3,0) {\eqref{non-FL-shift}};
\node[left] at (0,0) {\eqref{non-cMTM}};
\node[right] at (3.5,1.25) {\eqref{red-n}};
\node[right] at (-0.5,1.25) {\eqref{trans b}};
\end{tikzpicture}
\caption{\label{F1} Consistency of transformations and reductions.}
\end{figure}
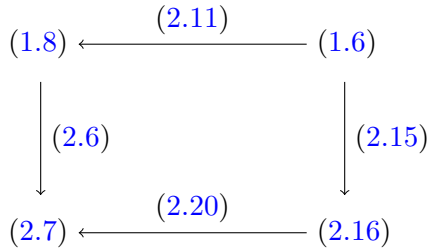

Such a consistency leads us to the following.
\begin{proposition}\label{prop-2-2}
If $u$ is a solution of the nonlocal FL equation  \eqref{non-FL-shift} and
$u(x,t)u(-x,-t) \to 0$ as $x\to \pm \infty$,
then
\begin{equation}\label{qu-non-p}
q_1=q \, e^{i \bar{\beta}}=u_x \, e^{i \bar{\beta}}, ~~ q_3= iu \, e^{i \bar{\beta}},
\end{equation}
where $\bar\beta$ is defined in \eqref{beta-b},
satisfy the nonlocal MTM \eqref{non-cMTM}.
\end{proposition}

Thus, to get solutions of the nonlocal MTM \eqref{non-cMTM},
we can either solve the unreduced MTM \eqref{31pKN-1} and then implement reductions,
or, making use of the connection \eqref{qu-non-p}, get   $q_1$ and $q_2$
from the solutions of the nonlocal FL equation.
We will focus on the first way in the next section.

\begin{remark}\label{Rem-2.1}
The transformation \eqref{qu} between the MTM and the FL equation
relies on the conservation law $(|q|^2)_t=(|u|^2)_x$
and the significance of $\beta=\partial^{-1}_x |q|^2$,
both of which require that $|u|\to 0$ as $x\to \pm \infty$.
For those solutions of the FL equation whose envelopes do not decrease
(e.g. $|u|$ with a periodic background, see Section 4.2.1 of \cite{LWZ-2022}
and Fig.\ref{F-2}(a), Fig.\ref{F-3}, Fig.\ref{F-4}(c,d) and Fig.\ref{F-5}(b) in this paper),
the transformation cannot be used to to generate solutions for the MTM.
Same remarks we have for the nonlocal MTM and nonlocal FL equation.
\end{remark}

\section{Solutions to the nonlocal  massive Thirring model}\label{sec-3}

In this section, we derive solutions for the nonlocal MTM \eqref{non-cMTM}
using the bilinearisation-reduction approach. We will bilinearize the unreduced MTM \eqref{31pKN-1},
present its double Wronskian solutions, and then implement reductions to get solutions for
the reduced equation \eqref{non-cMTM}.

\subsection{Bilinearisation-reduction approach}\label{sec-3-1}

\subsubsection{Double Wronskians }\label{sec-3-1-1}

We employ the notation $|\widehat{M-1};\widehat{N-1}|$ introduced in Ref.\cite{Nimmo-NLS-1983} to denote a $(M+N)\times(M+N)$ double Wronskian:
\begin{equation}\label{W}
|\widehat{M-1}; \widehat{N-1}|=|\W\phi^{(M-1)}; \W\psi^{(N-1)}|=|\phi,\partial_x \phi,\ldots,\partial_x^{M-1} \phi;
\psi,\partial_x \psi,\ldots,\partial_x^{N-1} \psi|,
\end{equation}
where $\phi$ and $ \psi $ are $(M+N)$-th order column vectors in the form
\begin{equation}
\phi=(\phi_1,\phi_2,\ldots,\phi_{M+N})^T,\ \psi=(\psi_1,\psi_2,\ldots,\psi_{M+N})^T,
\end{equation}
and the elements $\phi_j$ and $\psi_j$ are $C^{\infty}$ functions of $(x,t)$.
We also employ the short hands
\[W^{[M,N]}_{}(\phi; \psi)=|\W{M-1}; \W{N-1}|=|0,1,\cdots,M-1; 0,1,\cdots,N-1|
\]
and
\begin{equation}
\begin{array}{l}
|\t N; \W {M-1}|=|1,2,\cdots,N; 0,1,\cdots,M-1|,\\
|\W N; \t{M-1}|=|0,1,\cdots,N; 1,2,\cdots,M-1|,\\
|\b N; \W M|=|2,3,\cdots,N; 0,1,\cdots,M|,\\
|\t N; \t M|=|1,2,\cdots,N; 1,2,\cdots,M|.
\end{array}
\end{equation}

\subsubsection{Solutions of the unreduced equation \eqref{31pKN-1}}\label{sec-3-1-2}

By the following transformation
\begin{equation} \label{trans 5}
q_1=\frac{\t h}{f},\\
~~q_2=\frac{\b h}{s},\\
~~q_3=i\frac{g}{s},\\
~~q_4=-i\frac{h}{f},
\end{equation}
the unreduced MTM \eqref{31pKN-1} can be bilinearized as (cf.\cite{Pas-Spr-2023})
\begin{subequations}\label{4MTM-d}
\begin{align}
&B_1= D_{t}\t h\cdot s+gf=0, \label{23a}\\
&B_2= D_{t}\b h\cdot f+hs=0, \label{23b}\\
&B_3= D_{t}f\cdot s+igh=0, \label{23c}\\
&B_4= D_{x}g\cdot f-\t hs=0, \label{23d}\\
&B_5= D_{x}h\cdot s-\b hf=0, \label{23e}\\
&B_6= D_{x}f\cdot s-i\t h\b h=0, \label{23f}
\end{align}
\end{subequations}
where $D$ is  Hirota's bilinear operator $D$  defined as \cite{Hirota-1974}
\begin{equation*}
 D_x^m D_y^n f(x,y)\cdot g(x,y)=(\partial_x- \partial_{x'})^m(\partial_y- \partial_{y'})^n f(x,y)g(x',y')|_{x'=x,y'=y}.
\end{equation*}
In fact, by direct calculation one can find that
\begin{align*}
&\eqref{KN11}=\frac{i}{sf}B_1-\frac{i\t h}{s f^2}B_2=0,  \\
&\eqref{KN21}=\frac{i}{sf}B_3-\frac{i\b h}{s^2 f}B_2=0,  \\
&\eqref{KN31}=\frac{-1}{sf}B_4-\frac{g}{s^2 f }B_6=0,  \\
&\eqref{KN41}=\frac{1}{sf}B_5-\frac{i\b h}{s f^2}B_6=0,
\end{align*}
which confirms the bilinear forms in \eqref{4MTM-d}.
Note that two-soliton solutions of \eqref{31pKN-1} were obtained via the above bilinear form in \cite{Pas-Spr-2023}.
In the following, we present a compact form of $N$-soliton solutions in terms of double Wronskian.

\begin{theorem}\label{Theo-3-1}
The bilinear equations \eqref{4MTM-d} admit double Wronskian solutions
\begin{subequations}\label{wronskian-0}
\begin{align}
& f=|\widetilde{N}; \widehat{M-1}|,
\quad g=\frac{1}{2(2i)^{N-1}}|A||\widehat{N}; \widetilde{M-1}|,
\quad h=-i \,(2i)^{N-1}|A|^{-1}|\overline{N}; \widehat{M}|,\\
& s=|\widetilde{N}; \widetilde{M}|,
\quad \t h=-\frac{1}{(2i)^{N-1}}|A||\widehat{N}; \widehat{M-2}|,
\quad \b h=-\,(2i)^{N}|A|^{-1}|\b{N}; \widetilde{M+1}|,
\end{align}
\end{subequations}
where the elementary column vectors $\phi$ and $\psi$ satisfy
\begin{subequations}\label{wron-cond-x}
\begin{align}
& \phi_x=\frac{i}{2}A^{2}\phi,\quad \phi_t=-\frac{1}{4}\partial^{-1}_x \phi,\label{wron-cond-x-a}\\
& \psi_x=-\frac{i}{2}A^{2}\psi,\quad \psi_t=-\frac{1}{4}\partial^{-1}_x \psi,\label{wron-cond-x-b}
\end{align}
\end{subequations}
Here $A$ is an arbitrary invertible constant matrix in $\mathbb{C}_{(N+M)\times(N+M)}$.
Explicit forms of $\phi$ and $\psi$  satisfying \eqref{wron-cond-x} are
\begin{subequations}\label{phi-psi}
\begin{align}
& \phi=\exp\Bigl(\frac{i}{2}A^{2}x+\frac{i}{2}A^{-2}t\Bigr) C ,\label{phi}\\
& \psi=\exp\Bigl(-\frac{i}{2}A^{2}x-\frac{i}{2}A^{-2}t\Bigr) B,\label{psi}
\end{align}
\end{subequations}
where $B$ and $C$ are $(N+M)$-th order constant column vectors.
In addition, $A$ and any matrix to which is similar lead to same $\{q_j\}$ via the transformation \eqref{trans 5}.
\end{theorem}

The proof can be found in Appendix \ref{app-A}.

\subsubsection{Reductions}\label{sec-3-1-3}

We now derive solutions for the nonlocal MTM \eqref{non-cMTM}
by implementing reductions on the solutions we have obtained in Sec.\ref{sec-3-1-2}.
Such a reduction technique was developed in \cite{ChenDLZ-SAPM-2018,ChenZ-AML-2018,Deng-AMC-2018}
and has proved effective in deriving solutions for the nonlocal integrable equations
(see, e.g. \cite{ShiY-ND-2019,ZSL-2020,LWZ-ROMP-2022,WW-CNSNS-2022,ZLD-2023}
and a review \cite{ZDJ-review-2023}).

Let us present reductions of double Wronskians for getting solutions of the nonlocal MTM \eqref{non-cMTM}.

\begin{theorem}\label{Theo-3-3}

For the double Wronskians given in \eqref{wronskian-0} with
column vectors $\phi$ and $\psi$ defined in \eqref{phi-psi},
we impose a constraint $M=N$ and constraints on $\{A, B, C\}$ such that
\begin{subequations}\label{ASn}
\begin{align}
& A^2=\delta S^2,~~ \delta=\pm 1, \label{ASn1}\\
& B=SC,
\end{align}
\end{subequations}
where $S$ is a constant matrix in $\mathbb{C}_{2N\times 2N}$.
These settings yield a relation
\begin{equation}\label{MTM-psiSphic-n}
 \psi(x,t)=S\phi(-x,-t),
\end{equation}
and consequently the double Wronskians in \eqref{wronskian-0} satisfy
\begin{subequations}\label{fgh-n}
\begin{align}
& f(-x,-t)=  (-2i)^{N}\delta^{N}|S|^{-1} s(x,t),\\
& g(-x,-t)=  (-1)^{1-N}(2i)^{-N}\delta^{N+1}|S| h(x,t), \\
& \t h(-x,-t) = (-2i)^{N}\delta^{N+1}|S|^{-1}\b h(x,t). \label{fgh-n3}
\end{align}
\end{subequations}
Then, through the transformation \eqref{trans 5} we have
\begin{align}\label{q1-4n}
q_1(-x,-t)=\delta q_2(x,t), ~~ q_3(-x,-t)=\delta q_4(x,t), ~~~~ \delta =\pm 1,
\end{align}
which  agree with the nonlocal reduction \eqref{trans b}.

\end{theorem}

\begin{proof}

First, from the definition \eqref{phi-psi} and constraint \eqref{ASn}, we have
\begin{align*}
\psi(x,t)&= \exp{\Bigl(-\frac{i}{2}A^2x-\frac{i}{2}A^{-2} t\Bigr)}B \\
&=  \exp{\Bigl(-\frac{i}{2}(SA^2S^{-1})x-\frac{i}{2}(SA^{-2}S^{-1}) t\Bigr)}S C \\
&=  S\exp{\Bigl(-\frac{i}{2}A^2x-\frac{i}{2}A^{-2} t\Bigr)}C \\
&=  S\exp{\Bigl(\frac{i}{2}A^2(-x)+\frac{i}{2}A^{-2}(-t)\Bigr)}C \\
&= S \phi(-x,-t),
\end{align*}
which is \eqref{MTM-psiSphic-n}.
Then, we can express the double Wronskians in \eqref{wronskian-0} (with $N=M$)
in terms of $\phi$ as the following:
\begin{subequations}
\begin{align}
 &f(x,t)= |\widetilde{N}; \widehat{N-1}|
 =\Big(\frac{i}{2}\Big)^{N}|A^{2}\W{\phi^{[N-1]}_{}}(x,t)_{[x]}; S \W{\phi^{[N-1]}_{}}(-x,-t)_{[x]}|,\\
 &s(x,t)= |\widetilde{N}; \widetilde{N}|
  =(-1)^N \Big(\frac{i}{2}\Big)^{2N}|A|^{2}|\W{\phi^{[N-1]}_{}}(x,t)_{[x]};
  S \W{{\phi^{[N-1]}_{}}}(-x, -t)_{[x]}|,\\
 & \t h(x,t)=-\frac{1}{(2i)^{N-1}}|A||\widehat{N}; \widehat{N-2}|
 =-\frac{1}{(2i)^{N-1}}|A| |\W{\phi^{[N]}_{}}(x,t)_{[x]}; S \W{{\phi^{[N-2]}_{}}}(-x,-t)_{[x]}|,\\
& \b h(x,t)=-\,(2i)^{N}|A|^{-1}|\b{N}; \widetilde{N+1}|
 =\Big(\frac{i}{2}\Big)^{2N-1} |A||A^{2}\W{\phi^{[N-2]}_{}}(x,t)_{[x]};
S \W{{\phi^{[N]}_{}}}(-x,-t)_{[x]}|,\\
& g(x,t)=\frac{1}{2(2i)^{N-1}}|A||\widehat{N}; \widetilde{N-1}|
 =\frac{1}{2(2i)^{2N-2}}|A||\W{\phi^{[N]}_{}}(x,t)_{[x]};
 A^{2}S \W{{\phi^{[N-2]}_{}}}(-x,-t)_{[x]}|,\\
& h(x,t)=-i \,(2i)^{N-1}|A|^{-1}|\overline{N}; \widehat{N}|
 =i (-1)^{N}  \,\Big(\frac{i}{2}\Big)^{N-1}|A|^{-1}|A^{4}\W{\phi^{[N-2]}_{}}(x,t)_{[x]}; S\W{\phi^{[N]}_{}}(-x,-t)_{[x]}|,
\end{align}
\end{subequations}
where for convenience we have introduced short hand
\begin{align*}
  &\widehat{\phi^{[N]}}(a(x),b(t))_{[c(x)]} = \Bigl(\phi(a(x),b(t)),\partial_{c(x)}\varphi(a(x),b(t)),\partial_{c(x)}^2\varphi(a(x),b(t)),
  \cdots,\partial_{c(x)}^{N}\varphi(a(x),b(t))\Bigr)
\end{align*}
where $a(x)$ and $b(t)$ are functions of $x$ and $t$ respectively.
We find that
\begin{align*}
\t h(-x,-t)&
=-\frac{1}{(2i)^{N-1}}|A||\W{\phi^{[N]}_{}}(-x,-t)_{[-x]}; S \W{{\phi^{[N-2]}_{}}}(x,t)_{[-x]}|\\
& =\frac{1}{(2i)^{N-1}}|A||\W{\phi^{[N]}_{}}(-x,-t)_{[x]}; S \W{{\phi^{[N-2]}_{}}}(x,t)_{[x]}|\\
& =\frac{1}{(2i)^{N-1}}|A||S||S^{-1}\W{\phi^{[N]}_{}}(-x,-t)_{[x]}; \W{{\phi^{[N-2]}_{}}}(x,t)_{[x]}|\\
& = \frac{1}{(2i)^{N-1}}|A||S||\delta A^{-2}S\W{\phi^{[N]}_{}}(-x,-t)_{[x]}; \W{{\phi^{[N-2]}_{}}}(x,t)_{[x]}|\\
& =\frac{1}{(2i)^{N-1}}|S||A|^{-1}\delta^{N+1}|S\W{\phi^{[N]}_{}}(-x,-t)_{[x]}; A^2\W{{\phi^{[N-2]}_{}}}(x,t)_{[x]}|\\
& =\frac{1}{(2i)^{N-1}}|S||A|^{-1}\delta^{N+1}(-1)^{N+1}|A^2\W{{\phi^{[N-2]}_{}}}(x,t)_{[x]}; S\W{\phi^{[N]}_{}}(-x,-t)_{[x]}|\\
& =(-2i)^{N}\delta^{N+1}|S|^{-1}\b h(x,t),
\end{align*}
which is  relation  \eqref{fgh-n3}.
We can derive the first two relations in \eqref{fgh-n} in a similar way,
or one can also refer to Sec.3.2.2 of \cite{LWZ-2022} for a similar procedure.
Then, the reduction relations in \eqref{q1-4n} follow from \eqref{trans 5} and \eqref{fgh-n},
and we complete the proof.

\end{proof}

\begin{corollary}

Solutions of the nonlocal MTM \eqref{non-cMTM} are given by
\begin{equation} \label{trans 7}
q_1=\frac{\t h}{f},\\
~~q_3=i\frac{g}{s}=- i \delta \frac{h(-x,-t)}{f(-x,-t)},
\end{equation} 
where
\begin{subequations}\label{th-g-f-s-non}
\begin{align}
 & f(x,t)
 =(-2i)^{-N}|A^{2}\W{\phi^{[N-1]}_{}}(x,t)_{[x]}; S \W{\phi^{[N-1]}_{}}(-x,-t)_{[x]}|,\\ 
 & \t h(x,t)
 =-(2i)^{1-N}|A||\W{\phi^{[N]}_{}}(x,t)_{[x]}; S \W{{\phi^{[N-2]}_{}}}(-x,-t)_{[x]}|,\\
 & h(x,t)
 =-i (2i)^{1-N}  \,|A|^{-1}|A^{4}\W{\phi^{[N-2]}_{}}(x,t)_{[x]}; S\W{\phi^{[N]}_{}}(-x,-t)_{[x]}|,\\
 &  g(x,t)=(-1)^{N-1} 2^{1-2N}\,  |A||\W{\phi^{[N]}_{}}(x,t)_{[x]};
 A^{2}S \W{{\phi^{[N-2]}_{}}}(-x,-t)_{[x]}|,\\
  &s(x,t)=2^{-2N}\,|A|^{2}|\W{\phi^{[N-1]}_{}}(x,t)_{[x]};
  S \W{{\phi^{[N-1]}_{}}}(-x, -t)_{[x]}|,
 \end{align}
\end{subequations}
$\phi$ is defined by \eqref{phi}, and $A$ and $S$ obey the constraint \eqref{ASn1}.

\end{corollary}

\subsection{Explicit forms of $\phi$ and $\psi$ for the reduced equations}\label{sec-3-2}

Equation \eqref{ASn} is the constraints imposed on $A$ and $S$.
Solving them one can get explicit forms of $\phi$ and then get
explicit solutions for the nonlocal MTM.
Since \eqref{ASn} is the same as those for the nonlocal FL equations
(cf.\cite{LWZ-2022}), in the following we only list out main formulae.

\subsubsection{Case (1)}\label{sec-3-2-1}

Let $S=AT$  and to make \eqref{ASn} valid we assume that
\begin{equation}\label{TA2}
AT=TA,~~ T^2=\delta \mathbf{I}_{2N}^{},
\end{equation}
and assume $A$ and $T$ to be
block matrices
\begin{align}
T=\left( \begin{array}{cc} T_1 & T_2 \\T_3 & T_4 \\ \end{array}\right),~~
A=\left(\begin{array}{cc} K_1 & \mathbf{0}_N \\ \mathbf{0}_N & K_4 \\\end{array} \right),
\label{TA-block}
\end{align}
where $T_i$ and $K_i$ are $N\times N$ matrices.
A special solution of \eqref{TA2} can be given as
\begin{equation}\label{TA3}
 T_{1}=-T_{4}=\sqrt{\delta}\,\mathbf{I}_{N}, ~T_{2}=T_{3}=\mathbf{0}_{N},~
 K_{1}=\mathbf{K}_N\in \mathbb{C}_{N\times N},~ K_{4}=\mathbf{H}_{N}\in \mathbb{C}_{N\times N}.
\end{equation}
Explicit expression of $\phi$ can be given by
\begin{equation}\label{phi-pm}
\phi=\left( \begin{array}{c} \phi^+   \\ \phi^{-}\end{array}\right),
\end{equation}
where $\phi^{\pm}=(\phi^{\pm}_1, \phi^{\pm}_2, \cdots, \phi^{\pm}_N)^T$ take the forms
\begin{equation}\label{phi-pmm}
\phi^+=\exp\Big[\frac{i}{2}(K^{2}_1 x+K^{-2}_1 t)\Big]C^+,~~
\phi^-=\exp\Big[\frac{i}{2}(K^{2}_4 x+K^{-2}_4 t)\Big]C^-,
\end{equation}
and $C^{\pm}=(c^{\pm}_1, c^{\pm}_2, \cdots, c^{\pm}_N)^T$.

 When  $\mathbf{K}_{N}$ and $\mathbf{H}_{N}$ are complex matrices,
note that $\mathbf{K}_{N}$ and $\mathbf{H}_{N}$ are independent,
which means they are not necessary to be all diagonal or all of Jordan forms.
For $\phi^+$, when
\begin{equation}\label{KN-diag}
\mathbf{K}_{N}=\mathrm{D}[k_j]_{j=1}^{N}\doteq\mathrm{Diag}(k_1, k_2, \cdots, k_{N}),~~ k_j\in \mathbb{C},
\end{equation}
there is
\begin{align}\label{phi-m-KD}
\phi^+  =  (c^+_{1}\mathrm{e}^{\eta(k_{1})}, c^+_{2}\mathrm{e}^{\eta(k_{2})},\cdots,  c^+_{N}\mathrm{e}^{\eta(k_{N})} )^{T},~~ c^+_j\in \mathbb{C};
\end{align}
when
\begin{equation}\label{KN-jordan}
\mathbf{K}_{N}=
\mathbf{J}_{N}[k]\doteq \left(
  \begin{array}{cccc}
    k & 0 & \cdots & 0 \\
    1& k & \cdots & 0 \\
    \vdots & \ddots & \ddots & \vdots \\
    0 & \cdots& 1 & k \\
  \end{array}
\right)_{N\times N},~~ k\in\mathbb{C},
\end{equation}
$\phi^+$ takes the form
\begin{align}\label{phi-m-KJ}
\phi^+ =\Bigl(c^+ \mathrm{e}^{\eta(k)}, \partial_k(c^+ \mathrm{e}^{\eta(k)}),
\frac{1}{2!}\partial_k^{2}(c^+ \mathrm{e}^{\eta(k)}),
\cdots,  \frac{1}{(N-1)!}\partial_k^{N-1}(c^+ \mathrm{e}^{\eta(k)}) \Bigr)^{T},
\end{align}
where
\begin{align}
  \eta(k )=\frac{i}{2}(k^{2}x+ k^{-2}t),
  \label{eta}
\end{align}
and $c^+$ can be a function of $k$.
For $\phi^-$,
when $\mathbf{H}_{N}=\mathrm{D}[h_j]_{j=1}^{N},~h_j\in \mathbb{C}$,
$\phi^-$ takes the form of \eqref{phi-m-KD} with replacement
$(k_j, c_j^+)$ by $(h_j, c_j^-)$;
when $\mathbf{H}_{N}=\mathrm{J}_N[h],~h\in \mathbb{C}$,
$\phi^-$ takes the form of \eqref{phi-m-KJ} with replacement
$(k, c^+)$ by $(h, c^-)$.
Corresponding to \eqref{MTM-psiSphic-n}, $\psi$  takes a form
\begin{equation}\label{psi-phi-pmm-non}
\psi=\left( \begin{array}{r}\sqrt{\delta} \mathbf{K}_N^{}\,\phi^+(-x,-t)   \\
-\sqrt{\delta} \mathbf{H}_N\,\phi^{-}(-x,-t)\end{array}\right).
\end{equation}

\subsubsection{Case (2)}\label{sec-3-2-2}

The above solutions in Case (1) are based in the setting of \eqref{TA2} and \eqref{TA-block}.
There are outstanding solutions.
In fact,  equation \eqref{ASn} always has a solution
\begin{equation}
A=\sqrt{\delta}\, \t{\mathbf{I}} S
\end{equation}
for arbitrary  $A\in \mathbb{C}_{2N\times 2N}$,
where $\t{\mathbf{I}}$ is a certain matrix such that $\t{\mathbf{I}} S=S \t{\mathbf{I}}$ and
$\t{\mathbf{I}}^2$ is an identity matrix.
In practice, we can take
\begin{subequations}
\begin{align}
& A= \mathrm{Diag}(\mathrm{J}_{N_1}[k_1], \mathrm{J}_{N_2}[k_2], \cdots, \mathrm{J}_{N_s}[k_s]),\\
& \t{\mathbf{I}}=\mathrm{Diag}(\mathbf{I}_{N_1}, \mathbf{I}_{N_2}, \cdots, \mathbf{I}_{N_s}),
\end{align}
\end{subequations}
where $\sum_{j=1}^s N_j=2N$. The vector $\phi$ and $\psi$ in this case are given by
\begin{subequations}\label{phi-real}
\begin{equation}
\phi=(\phi_{N_1}[k_1], \phi_{N_2}[k_2],\cdots, \phi_{N_s}[k_s])^T,
~~ \psi=S\phi^*,
\end{equation}
where
\begin{align}
\phi_{N_j}[k] =\Bigl(c_j \mathrm{e}^{\eta(k)}, \partial_k(c_j \mathrm{e}^{\eta(k)}),
\frac{1}{2!}\partial_k^{2}(c_j \mathrm{e}^{\eta(k)}),
\cdots,  \frac{1}{(N_j-1)!}\partial_k^{N_j-1}(c_j \mathrm{e}^{\eta(k)}) \Bigr),
\end{align}
\end{subequations}
where $c_j$ can be a function of $k$.
The solutions resulting from this case can be understood as a type of partial-limit solutions
and can provide either algebraic solitons or high order ones  or their combinations \cite{W-2021,WW-2022}.

In addition, $\psi$  is taken as
\begin{equation}\label{psi-phi-pmm-non-2}
\psi(x,t)=S \phi(-x,-t).
\end{equation}

\section{Dynamics of solutions}\label{sec-4}

In this section, we discuss and  illustrate solutions of the nonlocal MTM \eqref{non-cMTM} with $\delta=1$.
In the following, let us look at the solutions case by case.

\subsection{Solutions from Case (1)}\label{sec-4-1}

\subsubsection{1-soliton solution}\label{sec-4-1-1}

When $\mathbf{K}_{N}$ takes the diagonal form as in \eqref{KN-diag}  and $\mathbf{H}_{N}=\mathrm{D}[h_j]_{j=1}^{N}$ with $h_j\in \mathbb{C}$,
we get solitons of the nonlocal MTM \eqref{non-cMTM}.
The simplest solution is the one-soliton solution (corresponding to $N=1$):
\begin{equation} \label{q1q3}
q_1=\frac{\t h}{f},\\
~~q_3=i\frac{g}{s},
\end{equation}
where
\begin{equation}
\t h=-{|\mathbf{K}_{1}||\mathbf{H}_{1}|}|\phi, \partial_x \phi|,
\quad f=|\partial_x \phi; \psi|,
\quad g={\frac{1}{2}|\mathbf{K}_{1}||\mathbf{H}_{1}|}|\phi, \partial_x \phi|,
\quad s=|\partial_x \phi; \partial_x \psi|,
\end{equation}
composed by
\begin{align*}
& \phi=(c_1 e^{\eta(k_1)}, d_1 e^{\eta(h_1)})^T, \\
& \psi=(c_1k_1 e^{-\eta(k_1)}, -d_1h_1 e^{-\eta(h_1)})^T,
\end{align*}
where $\eta(k)$ is given in \eqref{eta}, $k_1, h_1, c_1, d_1\in \mathbb{C}$.

Explicit formulae for $q_1$ and $q_3$ are
\begin{subequations} \label{1ss-q1q3}
\begin{align}
q_1 & = \frac{h_{1}^{2}-k_{1}^{2}}
{k_{1}\, \mathrm{e}^{-i(h_{1}^{2}x+\frac{t}{h_{1}^{2}})} +h_{1}\, \mathrm{e}^{-i(k_{1}^{2}x +\frac{t}{k_{1}^{2}})}},\\
q_3 & =\frac{h_{1}^{2}-k_{1}^{2}}
{k_{1}h_{1} \left[k_{1}\, \mathrm{e}^{-i(k_{1}^{2}x+\frac{t}{k_{1}^{2}})} +h_{1}\, \mathrm{e}^{-i(h_{1}^{2}x +\frac{t}{h_{1}^{2}})} \right]}.
\end{align}
\end{subequations}
The envelope of $q_1$ reads
\begin{equation}\label{1-sol-4-MTM-q1}
|q_1|^2= \frac{(a_{1}^{2}-b_{1}^{2}-m_{1}^{2}+s_{1}^{2})^{2}+4(a_{1}b_{1}-m_{1}s_{1})^{2}}
{2|k_1||h_1| \, \mathrm{e}^{2W_{1}}\left[\cosh \left(2W_{2}+\ln \frac{|h_1| }{|k_1| }\right)
+ \sin(W_{3}+\omega_{1})\right]},
\end{equation}
and for $q_3$ we have
\begin{equation}\label{1-sol-4-MTM-q3}
|q_3|^2=  \frac{(a_{1}^{2}-b_{1}^{2}-m_{1}^{2}+s_{1}^{2})^{2}+4(a_{1}b_{1}-m_{1}s_{1})^{2}}
{2|k_1|^3|h_1|^3 \, \mathrm{e}^{2W_{1}}\left[\cosh \left(2W_{2}+\ln \frac{|h_1| }{|k_1| }\right)
+ \sin(\omega_{1}-W_{3})\right]},
\end{equation}
where we have taken $k_{1}=a_{1}+ib_{1}$, $h_1=m_1+is_1$, $a_1,b_1,m_1,s_1\in \mathbb{R}$, and where
\begin{subequations}
\begin{align}
&W_{1}=(a_{1}b_{1}+m_{1}s_{1})x
-\left(\frac{a_{1}b_{1}}{|k_1|^{4} }+\frac{m_{1}s_{1}}{|h_1|^{4}}\right) t,\\
&W_{2}=(a_{1}b_{1}-m_{1}s_{1})x
-\left(\frac{a_{1}b_{1}}{|k_1|^{4} }-\frac{m_{1}s_{1}}{|h_1|^{4}}\right) t, \\
&W_{3}=(a_{1}^{2}-b_{1}^{2}-m_{1}^{2}+s_{1}^{2})x
+\left(\frac{a_{1}^{2}-b_{1}^{2}}{|k_1|^{4}}
-\frac{m_{1}^{2}-s_{1}^{2}}{|h_1|^{4}}\right) t,\\
& \omega_{1}=\arctan \frac{a_{1}m_{1}+b_{1}s_{1}}{a_{1}s_{1}-b_{1}m_{1}}.
\end{align}
\end{subequations}

It can be found that
$W_1\equiv 0$ if
\begin{subequations}
\begin{align}
 (a_1,b_1)=(\pm m_1, \mp s_1),~ \mathrm{or} ~ (a_1,b_1)=(\pm s_1, \mp m_1),\label{W1=0-a}
\end{align}
or
\begin{align}
 a_1b_1=m_1s_1=0,~\mathrm{but}~ |k_1||h_1|\neq 0; \label{W1=0-b}
\end{align}
\end{subequations}
$W_2\equiv 0$ if \eqref{W1=0-b} holds, or
\begin{equation}
 (a_1,b_1)=( m_1,  s_1),~ \mathrm{or} ~ (a_1,b_1)=( s_1,  m_1); \label{W2=0-a}
 \end{equation}
$W_3\equiv 0$ if
\begin{equation}
 (a_1^2,b_1^2)=( m_1^2,  s_1^2). \label{W3=0-a}
 \end{equation}

With these in hand, we may obtain desired solutions by arranging real parameters $a,b,m,s$.
For example, if we take $a_1=s_1=0$ but $|k_1|\neq |h_1| \neq 0$ so that \eqref{W1=0-b} holds,
we get $W_{1}=W_{2}\equiv 0$ and
\begin{equation}\label{1ss-p-non}
|q_1|^2=\frac{ (b_1^2+m_1^2)^{2}}{
b_1^2+m_1^2- 2|b_1m_1| \sin \left((b_1^2+m_1^2)(x+\frac{t}{b_1^2 m_1^2})\right)},
\end{equation}
which is a nonsingular periodic wave (by virtue of $|b_1|\neq |m_1|$, i.e. $|k_1|\neq |h_1|$).
This wave is depicted in Fig.\ref{F-2}(a).

When $W_{1}\equiv W_{3}\equiv 0$ but $W_{2}\neq 0$, which can hold by taking, e.g. $(a_1,b_1)=(m_1,-s_1)$,
we get 1-soliton solution (1SS), which yields
\begin{equation}\label{1ss-s-non}
|q_1|^2=\frac{8a_1^2b_1^2}{(a_1^2+b_1^2)
(\cosh 4 a_1b_1 W'_2+ \sin \omega'_1)},
\end{equation}
where
$W'_2= x-\frac{t}{(a_1^2+b_1^2)^{2}}$ and $\omega'_1=\arctan\frac{b_1^2-a_1^2}{2a_1b_1}$.
It is depicted in Fig.\ref{F-2}(b).

\captionsetup[figure]{labelfont={bf},name={Fig.},labelsep=period}
\begin{figure}[ht]
\centering
\subfigure[ ]{
\begin{minipage}[t]{0.44\linewidth}
\centering
\includegraphics[width=2.1in]{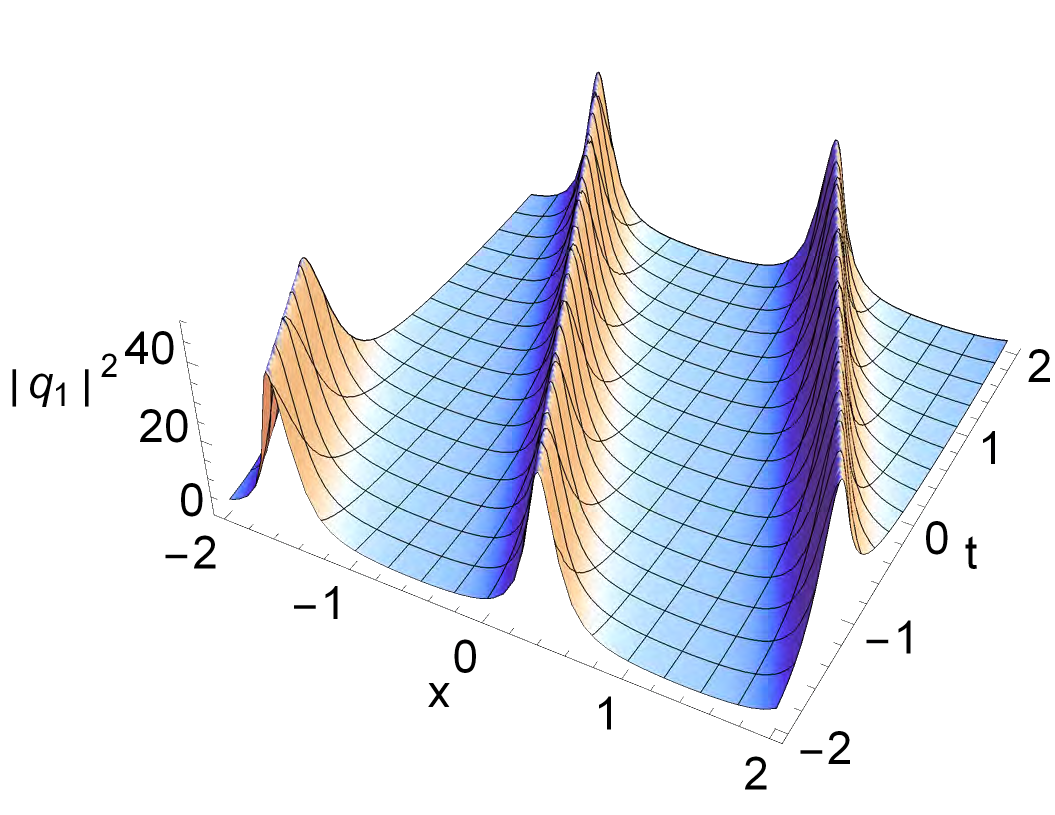}
\end{minipage}%
}%
\subfigure[ ]{
\begin{minipage}[t]{0.44\linewidth}
\centering
\includegraphics[width=2.1in]{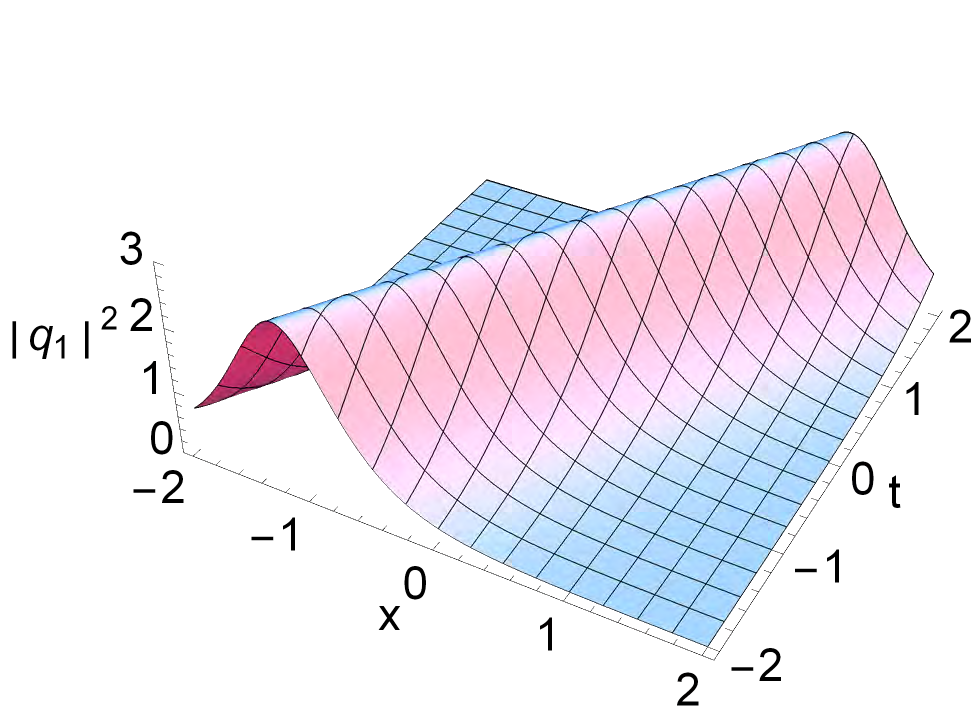}
\end{minipage}%
}
\caption{Shape and motion of 1SS of the nonlocal  MTM  \eqref{non-cMTM}.
(a) Periodic wave given by \eqref{1ss-p-non} with $k_1=i$ and $h_1=1.5$.~
(b) Soliton given by \eqref{1ss-s-non} with $k_1=0.8+0.8i$ and $h_1=0.8-0.8i$.  }
\label{F-2}
\end{figure}

There exist kink-type waves but always with singularities.
Considering the case that there is only one number being zero among $(a_1,b_1,m_1,s_1)$,
e.g. only $m_1=0$, i.e.
\begin{equation}\label{abms}
(a_1,b_1,m_1,s_1)=(a_1,b_1,0,s_1),~ \mathrm{and} ~ a_1b_1s_1\neq 0,
\end{equation}
we have
\begin{subequations}\label{W1-omega}
\begin{align}
& W_1=W_2=a_1b_1\left(x-\frac{t}{|k_1|^4}\right),\\
& W_3=(a_1^2-b_1^2+s_1^2)x+\left(\frac{a_1^2-b_1^2}{|k_1|^4}+\frac{1}{s_1^2}\right)t,~~
\omega_1=\arctan\frac{b_1}{a_1}.
\end{align}
\end{subequations}
It is easy to check that the slopes of lines $W_1=0$ and $W_3=0$ can never be same
in case choosing \eqref{abms}.
In this case, \eqref{1-sol-4-MTM-q1} turns out to be
\begin{equation}\label{1ss-k-non}
|q_1|^2=
\frac{(a_{1}^{2}-b_{1}^{2}+s_{1}^{2})^{2}+4(a_{1}b_{1})^{2}}
{ |h_1|^2 y^2+ |k_1|^2+2 y |k_1||h_1|\sin z},
\end{equation}
where
\[y=e^{2W1},~ z=W_3+\omega_1,\]
and $W_1, ~\omega_1$ take the forms in \eqref{W1-omega}.
This is a kink-type wave for any given $t$: when $x\to \pm\infty$,
$|q_1|^2$ goes to zero on one side and
$\frac{(a_{1}^{2}-b_{1}^{2}+s_{1}^{2})^{2}+4(a_{1}b_{1})^{2}}{|k_1|^{2}}$
on the other side, or the other way around, depending on sgn$[a_1b_1]$.
However, there are infinitely many poles appearing at the intersections
\[\left\{\begin{array}{l}
W_1=\frac{1}{2}\ln \frac{|k_1|}{|h_1|},\\
W_3+\omega_1=2j\pi-\frac{\pi}{2},~~ j\in \mathbb{Z},
\end{array}\right.
\]
and all poles are located at the line $W_1=\frac{1}{2}\ln \frac{|k_1|}{|h_1|}$.
Such a solution is illustrated in Fig.\ref{F-3}.

\captionsetup[figure]{labelfont={bf},name={Fig.},labelsep=period}
\begin{figure}[h]
\centering
\begin{minipage}[t]{0.45\linewidth}
\centering
\includegraphics[width=2.3in]{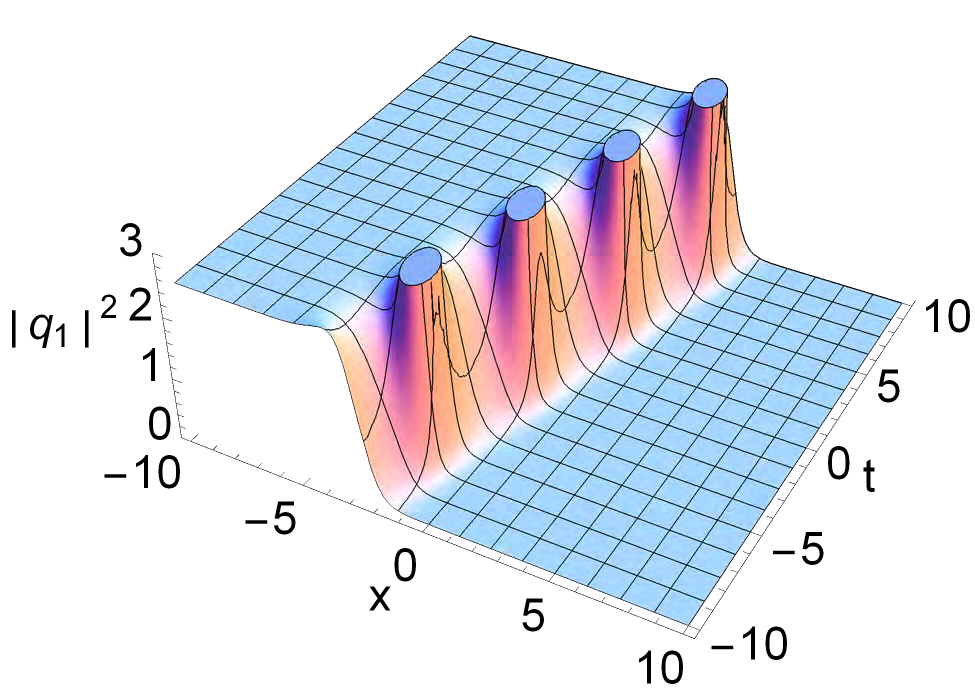}
\end{minipage}%
\caption{ Shape and motion of a kink-type wave of the nonlocal MTM   \eqref{non-cMTM},
given by \eqref{1ss-k-non}  with $k_1=1+i$ and $h_1=i$.}
\label{F-3}
\end{figure}

For $q_3$ we can have similar analysis and we skip its illustrations.

\subsubsection{Two-soliton solution and double-pole solution}\label{sec-4-1-2}

Two-soliton solution (2SS) is obtained with $N=2$. In this case, we have
\begin{equation} \label{2ss-q1q3}
q_{1 (2\mathrm{SS})}=\frac{\t h}{f},\\
~~q_{3 (2\mathrm{SS})}=i\frac{g}{s},
\end{equation}
where
\begin{subequations}
\begin{align}
& \t h=-\frac{1}{2i}{|\mathbf{K}_{2}||\mathbf{H}_{2}|}|\phi, \partial_x\phi, \partial^2_x\phi; \psi|,
~~f=|\partial_x\phi, \partial^2_x\phi; \psi, \partial_x\psi|, \\
& g=\frac{1}{4i}{|\mathbf{K}_{2}||\mathbf{H}_{2}|}|\phi, \partial_x \phi, \partial^2_x\phi; \partial_x\psi|,
~~s=|\partial_x \phi, \partial^2_x\phi; \partial_x \psi, \partial^2_x\psi|.
\end{align}
\end{subequations}
A particular case is of both $\mathbf{K}_N$ and  $\mathbf{H}_N$ being diagonal, in which we have
\begin{subequations}\label{a}
\begin{align}
& \phi=(c_1e^{\eta(k_1)}, c_2e^{\eta(k_2)}, d_1e^{\eta(h_1)}, d_2e^{\eta(h_2)})^T,\\
& \psi=(k_1c_1e^{-\eta(k_1)}, k_2c_2e^{-\eta(k_2)}, -h_1d_1e^{-\eta(h_1)}, -h_2 d_2e^{-\eta(h_2)})^T.
\end{align}
\end{subequations}
This case leads to usual 2SS.
In addition,
when both $\mathbf{K}_N$ and $\mathbf{H}_N$ are Jordan blocks, we have
\begin{subequations}\label{b}
\begin{align}
& \phi=\left(c_1 e^{\eta(k_1)}, c_1 \partial_{k_1} e^{\eta(k_1)},
d_1 e^{\eta(h_1)}, d_1\partial_{h_1} e^{\eta(h_1)}\right)^T,\\
& \psi\!=\!\left(k_1c_1e^{-\eta(k_1)}, c_1 e^{-\eta(k_1)}\!+\!c_1 k_1 \partial_{k_1} e^{-\eta(k_1)},
 -h_1d_1e^{-\eta(h_1)},- d_1e^{-\eta(h_1)}\!-\!h_1 d_1\partial_{h_1} e^{-\eta(h_1)}\right)^T;
\end{align}
\end{subequations}
and when  $\mathbf{K}_N$ is diagonal and $\mathbf{H}_N$ is a Jordan block, we have
\begin{subequations}\label{c}
\begin{align}
& \phi=\left(c_1 e^{\eta(k_1)}, c_2 e^{\eta(k_2)},
d_1 e^{\eta(h_1)}, d_1\partial_{h_1} e^{\eta(h_1)}\right)^T,\\
& \psi=\left(k_1c_1e^{-\eta(k_1)}, k_2 c_2 e^{-\eta(k_2)},
 -h_1d_1e^{-\eta(h_1)},- d_1e^{-\eta(h_1)}-h_1 d_1\partial_{h_1} e^{-\eta(h_1)}\right)^T.
\end{align}
\end{subequations}
Here $\eta$ is defined by \eqref{eta}, $k_j, h_j, c_j, d_j\in \mathbb{C}$.
This case   \eqref{b} gives rise to double-pole solutions,
while case \eqref{c} leads to mixed solutions.

Two-soliton interactions are illustrated in Fig.\ref{F-4} and  Fig.\ref{F-5},
from which we can see that  the interactions of 2SS
are more complicated in nonlocal case (see also the nonlocal Gross-Pitaevskii equation
\cite{Liu-ROMP-2020}).

\captionsetup[figure]{labelfont={bf},name={Fig.},labelsep=period}
\begin{figure}[h]
\centering
\subfigure[ ]{
\begin{minipage}[t]{0.42\linewidth}
\centering
\includegraphics[width=2.1in]{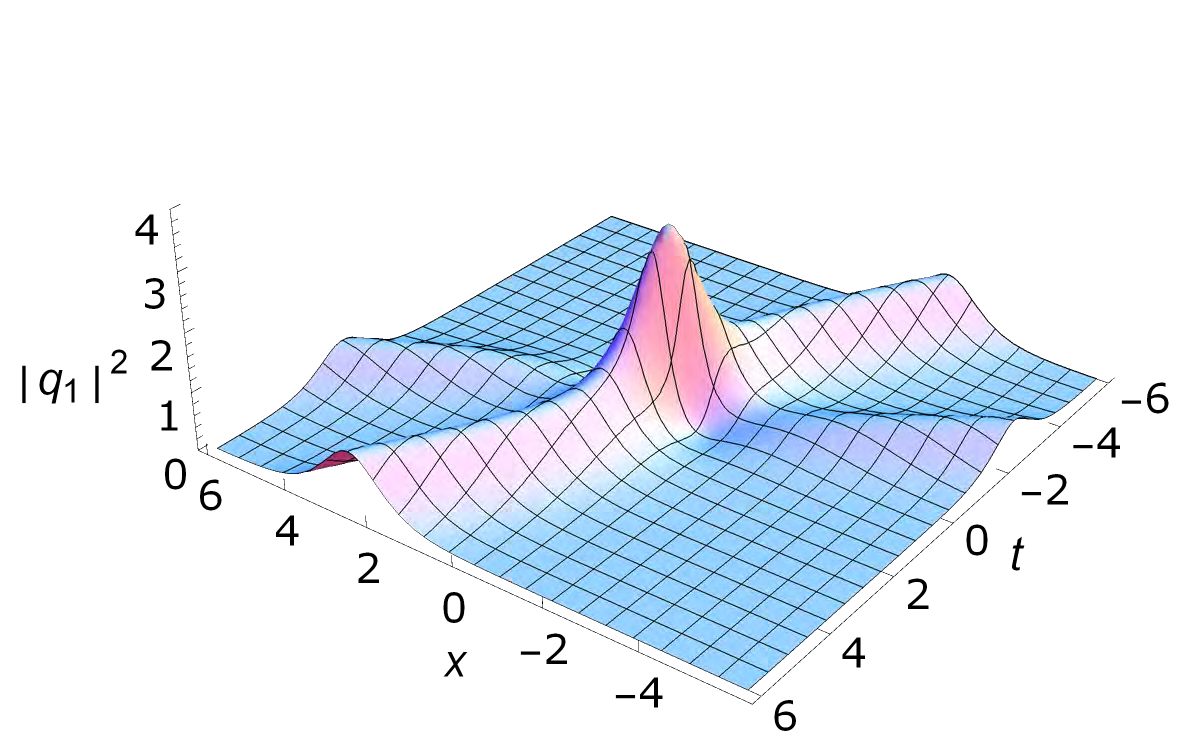}
\end{minipage}%
}%
\subfigure[ ]{
\begin{minipage}[t]{0.42\linewidth}
\centering
\includegraphics[width=2.1in]{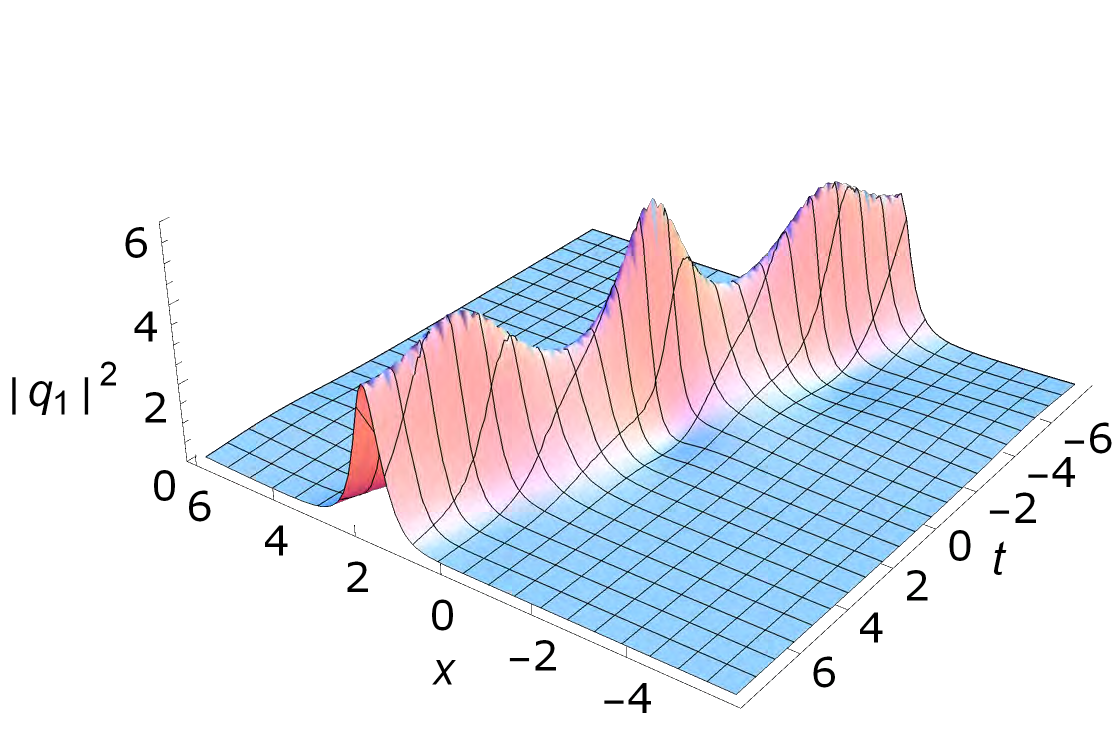}
\end{minipage}%
}\\
\subfigure[ ]{
\begin{minipage}[t]{0.42\linewidth}
\centering
\includegraphics[width=2.1in]{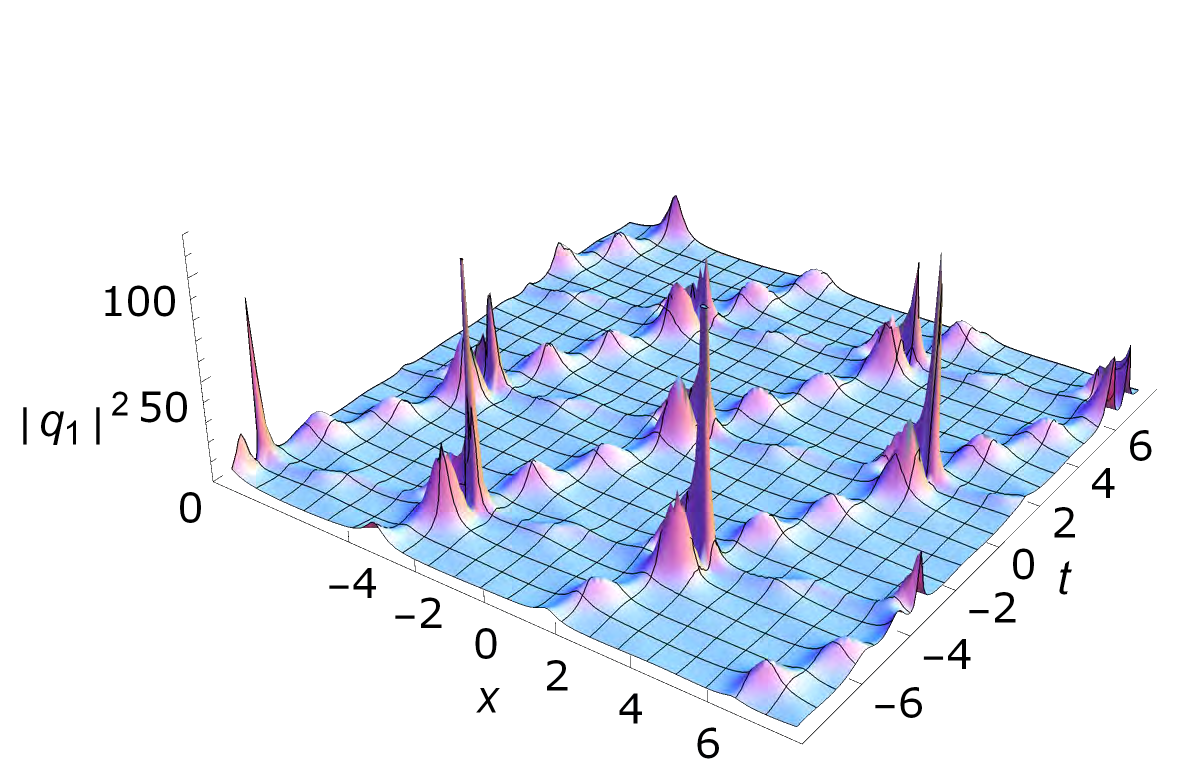}
\end{minipage}%
}%
\subfigure[ ]{
\begin{minipage}[t]{0.42\linewidth}
\centering
\includegraphics[width=2.1in]{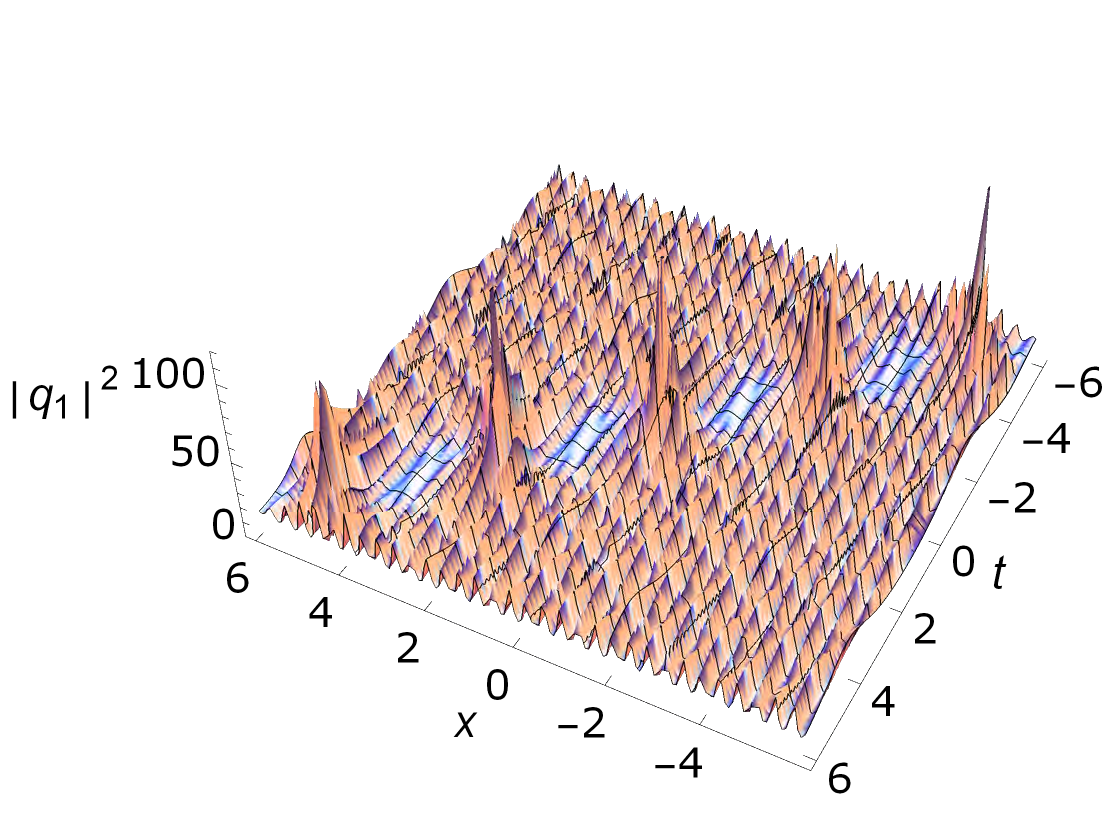}
\end{minipage}%
}
\caption{Shape and motion of 2SS of the nonlocal MTM   \eqref{non-cMTM}.~
Envelope $|q_1|^2$ of \eqref{2ss-q1q3} with \eqref{a} for
(a) $k_1=1+0.5i$, $h_1=1-0.5i$,  $k_2=0.8+0.4i$,
$h_2=0.8-0.4i$ and $c_1=c_2=d_1=d_2=1$; ~
(b)   $k_1=1+i$, $h_1=1-i$,  $k_2=1+0.2i$,
$h_2=1-0.2i$ and $c_1=c_2=d_1=d_2=1$; ~
(c)   $k_1=0.6i$, $h_1=-1.5$,  $k_2=0.8i$,
$h_2=-1$ and $c_1=c_2=d_1=d_2=1$;
~
(d)   $k_1=0.8+0.6i$, $h_1=0.8-0.6i$,  $k_2=i$,
$h_2=-4$ and $c_1=c_2=d_1=d_2=1$.
}
\label{F-4}
\end{figure}

\captionsetup[figure]{labelfont={bf},name={Fig.},labelsep=period}
\begin{figure}[h]
\centering
\subfigure[ ]{
\begin{minipage}[t]{0.40\linewidth}
\centering
\includegraphics[width=2.1in]{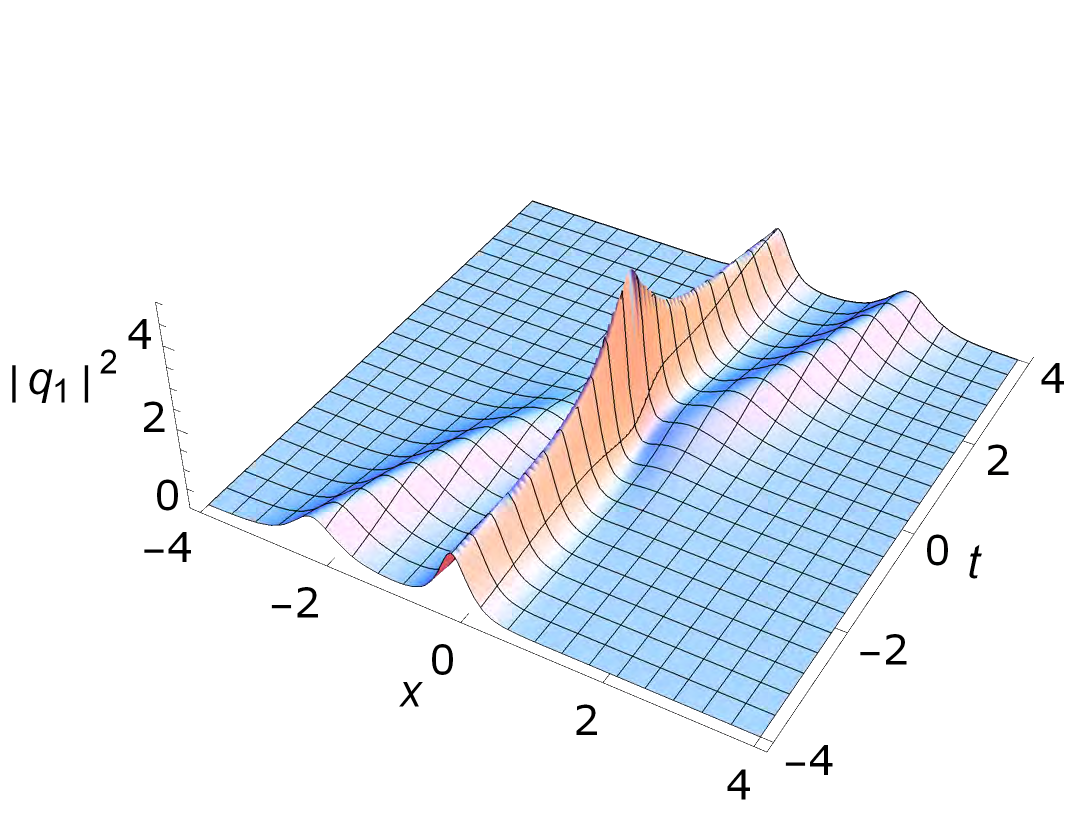}
\end{minipage}%
}%
\subfigure[ ]{
\begin{minipage}[t]{0.40\linewidth}
\centering
\includegraphics[width=2.1in]{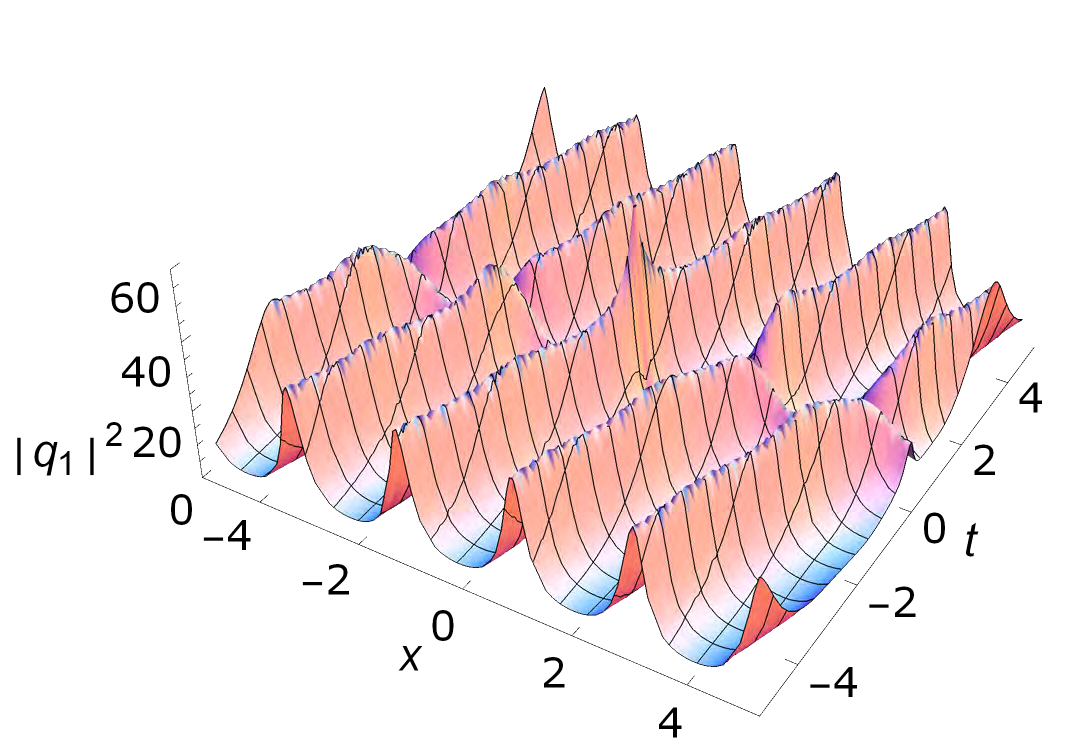}
\end{minipage}%
}%
\caption{Shape and motion of double-pole and mixed solutions  of the nonlocal MTM  \eqref{non-cMTM}.~
(a) Envelope $|q_1|^2$ of \eqref{2ss-q1q3} with \eqref{b}
in which $k_1=1+i$, $h_1=1-i$, $k_2=1$, $h_2=1$ and $c_1=d_1=1$. ~
(b) Envelope $|q_1|^2$ of \eqref{2ss-q1q3} with \eqref{c}
in which $k_1=1$, $k_2=2$, $h_1=-0.5$, $h_2=1$
and $c_1=c_2=d_1 =1$. }
\label{F-5}
\end{figure}

\subsection{Solutions from Case (2)}\label{sec-4-2}
\subsubsection{Algebraic soliton solutions}\label{sec-4-2-1}

Based on \textbf{Case (2)}, to generate algebraic soliton solutions, we consider
\begin{subequations}\label{AT-alg}
        \begin{equation}\label{AT-diag-jordon}
         A= \mathrm{Diag}(A_1,A_2,\cdots,A_N),~~
         T= \mathrm{Diag}(T_1,T_2,\cdots,T_N),~~
         S=AT,
         \end{equation}
        where
        \begin{equation}\label{Aj-Tj-alg}
          A_j=\left(
            \begin{array}{cc}
              k_j&0\\
              1& k_j
            \end{array}
            \right),~~
             T_j=\left(
              \begin{array}{cc}
                \sqrt{\delta}&0 \\
               0& \sqrt{\delta}
              \end{array}
              \right),~~k_j\in\mathbb{R},~~j=1,\cdots,N,
        \end{equation}
        \end{subequations}
and the corresponding vector $\phi$ and $\psi$  are given by
          \begin{align}\label{phi-alg}
            \phi  =\Bigl(\mathrm{e}^{\eta (k_1)},\partial_{ k_1}\mathrm{e}^{\eta(k_1)},\mathrm{e}^{\eta (k_2)},\partial_{ k_2}\mathrm{e}^{\eta(k_2)},
               \cdots,\mathrm{e}^{\eta (k_N)},\partial_{k_N}\mathrm{e}^{\eta(k_N)}\Bigr)^{T},
          ~~ \psi=S\phi(-x,-t),
           \end{align}
where
\begin{equation}
\eta(k_j )=\frac{i}{2}(k_j^{2}x+ k_j^{-2}t),~ k_j\in\mathbb{R}, ~~\delta=\pm 1,
\end{equation}
and we have taken
         $C=(1,1,\cdots,1)^T$ for convenience.

The simplest case is of $N=1$, where
\begin{equation}\label{A-T-alg-1ss}
  A=\left(
    \begin{array}{cc}
      k_1&0\\
      1& k_1
    \end{array}
    \right),~~
     T=\left(
      \begin{array}{cc}
        \sqrt{\delta}&0 \\
       0& \sqrt{\delta}
      \end{array}
      \right),
\end{equation}
and
\begin{align}
  \phi  =\Bigl(\mathrm{e}^{\eta (k_1)},\partial_{ k_1}\mathrm{e}^{\eta(k_1)}\Bigr)^{T},
~~ \psi=S\phi(-x,-t)=AT\phi(-x,-t).
\end{align}
Then the resulting  solutions (via \eqref{trans 7},\eqref{th-g-f-s-non}) of the nonlocal MTM \eqref{non-cMTM} are
\begin{subequations}
  \begin{align}
    &q_1=-\frac{2 i k_1^3 e^{i( k_1^{2}x+ k_1^{-2}t)}}{\sqrt{\delta} (2 k_1^4 x-i k_1^2-2 t)},\\
 & q_3=-\frac{2i k_1 e^{i( k_1^{2}x+ k_1^{-2}t)}}{\sqrt{\delta} (2 k_1^4 x+i k_1^2-2 t)},
  \end{align}
\end{subequations}
and the corresponding envelopes read
\begin{subequations}
  \begin{align}
    &|q_1|^2=\frac{4 k_1^6}{(2k_1^4x-2t)^2+k_1^4},\label{q1-alg-1ss}\\
 & |q_3|^2=\frac{4 k_1^2}{(2k_1^4x-2t)^2+k_1^4}.\label{q3-alg-1ss}
  \end{align}
\end{subequations}
These are nonsingular algebraic solitons as they look like solitons but with algebraic decaying (for given $x$ or $t$).
Both the algebraic solitons generated by $|q_1|^2$ and $|q_3|^2$ travel along $x=k_1^{-4}t$, with amplitudes $4k_1^2$ and $4k_1^{-2}$, respectively,
see Fig.\ref{F-6}.
\captionsetup[figure]{labelfont={bf},name={Fig.},labelsep=period}
\begin{figure}[ht]
\centering
\subfigure[ ]{
\begin{minipage}[t]{0.45\linewidth}
\centering
\includegraphics[width=2.5in]{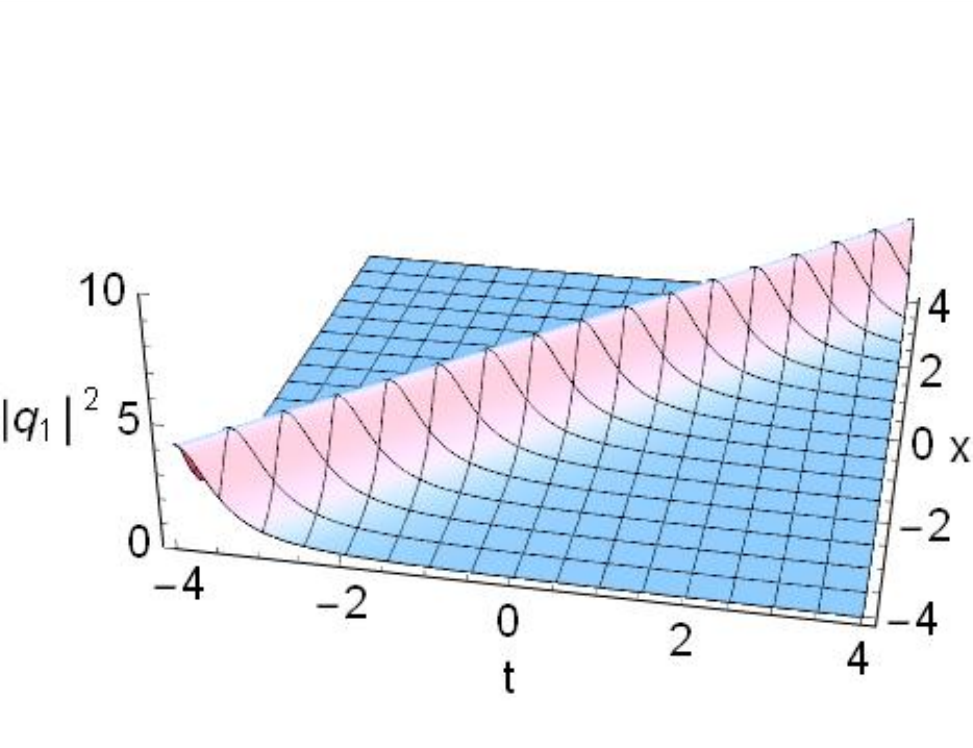}
\end{minipage}%
}%
\subfigure[ ]{
\begin{minipage}[t]{0.45\linewidth}
\centering
\includegraphics[width=2.5in]{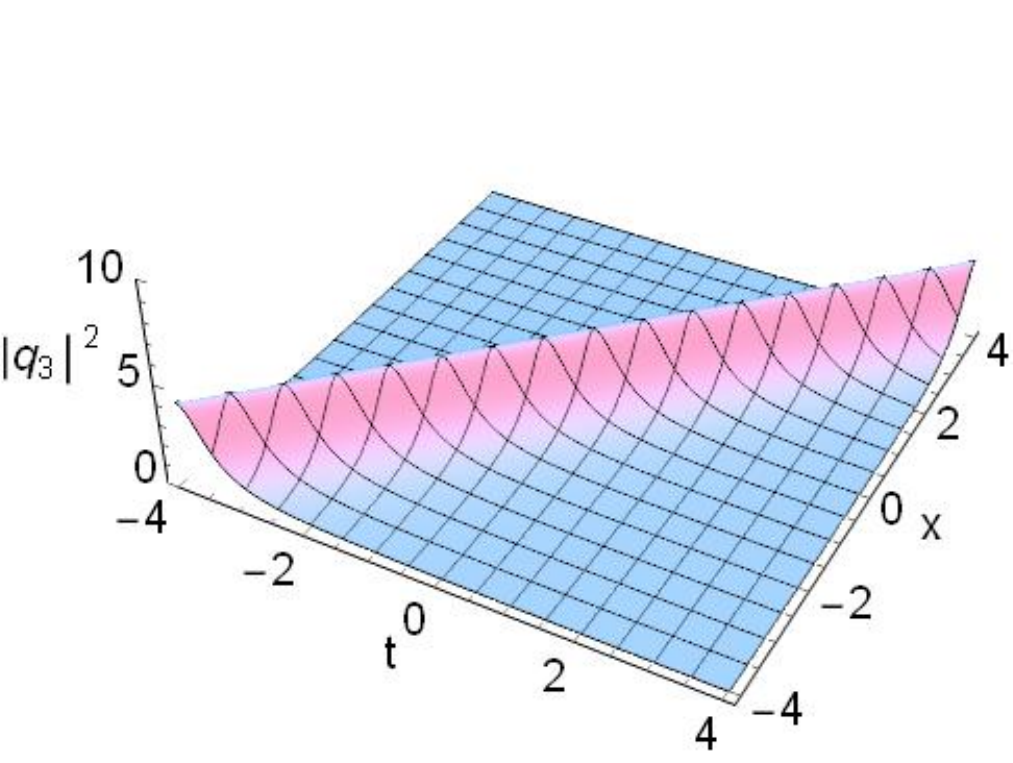}
\end{minipage}%
}
\subfigure[ ]{
\begin{minipage}[t]{0.45\linewidth}
\centering
\includegraphics[width=2.0in]{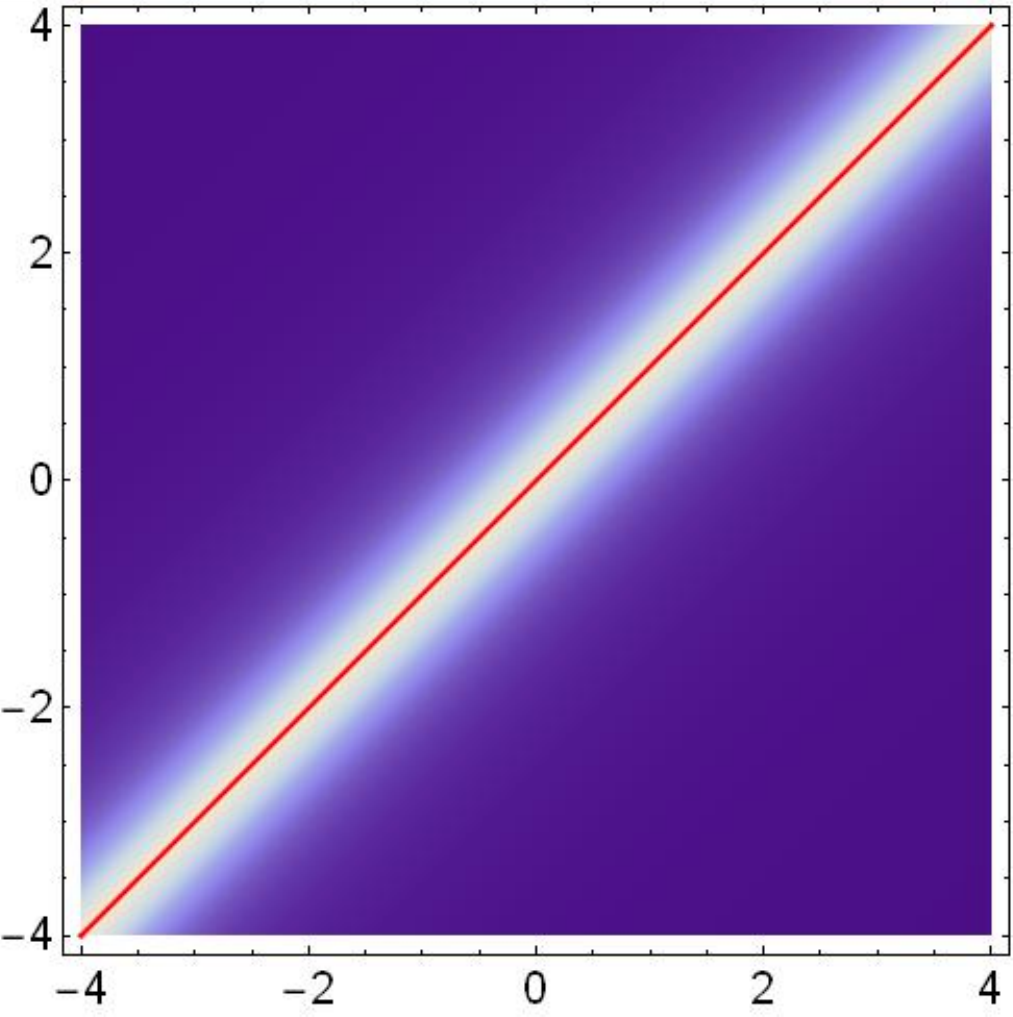}
\end{minipage}%
}
\subfigure[ ]{
\begin{minipage}[t]{0.45\linewidth}
\centering
\includegraphics[width=2.0in]{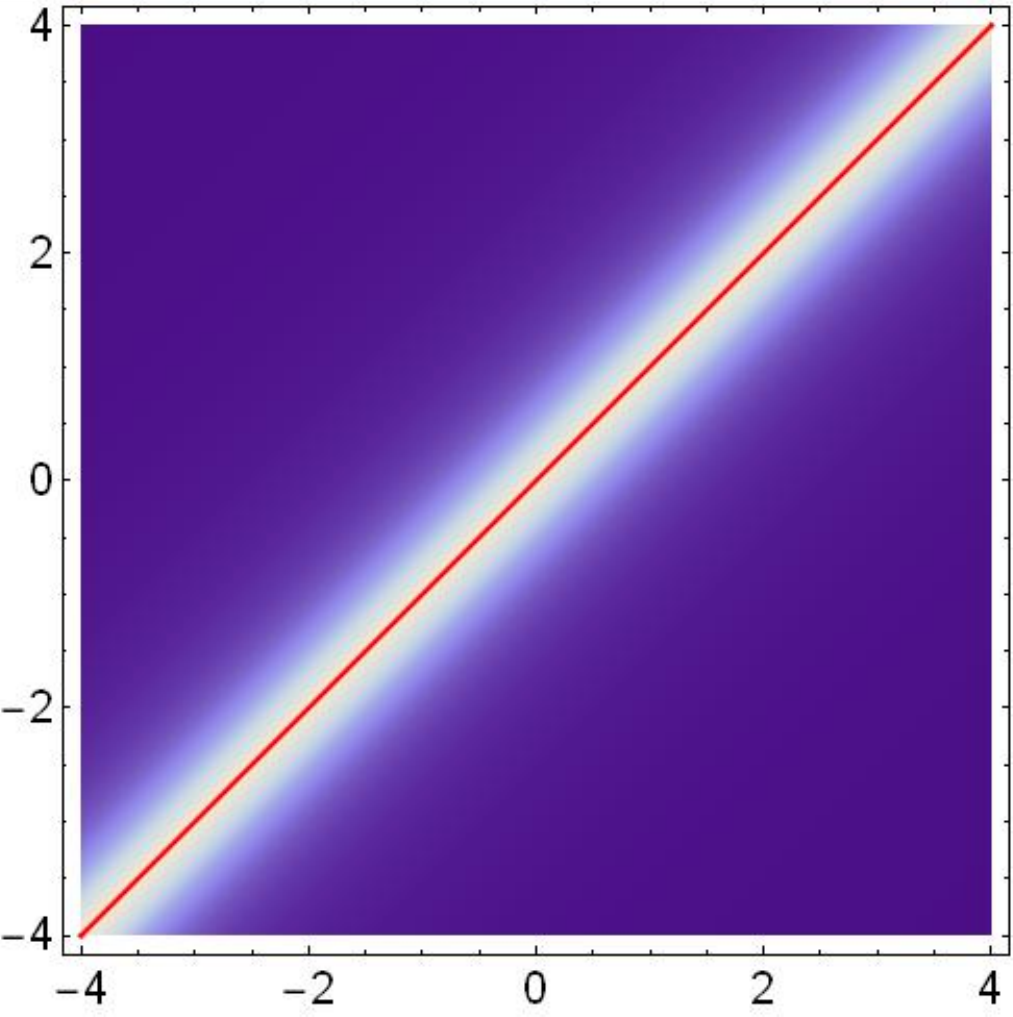}
\end{minipage}%
}
\caption{Shape and motion of algebraic 1SS of the nonlocal  MTM   \eqref{non-cMTM}.
(a) Envelope $|q_1|^2$ given in \eqref{q1-alg-1ss} with $\delta=1,~k_1=1$.
(b) Envelope $|q_3|^2$ given in \eqref{q3-alg-1ss} with $\delta=1,~k_1=1$.
(c) Density plot of (a) where the red line is $x=k_1^{-4}t$.
(d) Density plot of (b) where the red line is $x=k_1^{-4}t$.
 }
\label{F-6}
\end{figure}

Next, we consider the case of $N=2$, where we have
\begin{equation}
  A=\left(
   \begin{array}{cccc}
    k_1 &  0 & 0 & 0 \\
     1 & k_1&0 & 0 \\
     0 &0 & k_2&  0 \\
     0 & 0 &1 & k_2
   \end{array}
 \right),~~
 T= \left(
  \begin{array}{cccc}
   \sqrt{\delta}& 0 & 0 &  0 \\
    0&  \sqrt{\delta} & 0 &  0 \\
   0& 0&  \sqrt{\delta}&  0 \\
   0& 0& 0 &  \sqrt{\delta} \\
  \end{array}
\right),
 \end{equation}
and
\begin{align}
  \phi  =\Bigl(\mathrm{e}^{\eta (k_1)},\partial_{ k_1}\mathrm{e}^{\eta(k_1)},\mathrm{e}^{\eta (k_2)},\partial_{ k_2}\mathrm{e}^{\eta(k_2)},\Bigr)^{T},
~~ \psi=S\phi(-x,-t)=AT\phi(-x,-t).
 \end{align}
We skip the expressions of $q_1$ and $q_3$, and instead,  present the explicit formulae of envelops
\begin{subequations}
  \begin{align}
  |q_1|^2&=\frac{\mathcal{A}_1}{\mathcal{B}_1},\label{q1-alg-2ss}\\
 |q_3|^2&=\frac{\mathcal{A}_2}{\mathcal{B}_2},\label{q3-alg-2ss}
\end{align}
  where
  \begin{align}
 \mathcal{A}_1=&4(k_1^2-k_2^2)^2\Big((M_1\cos{\theta_1}-M_2\sin{\theta_1}
 +M_3\cos{\theta_2}-M_4\sin{\theta_2})^2\nonumber\\
 &+(M_1\sin{\theta_1}+M_2\cos{\theta_1}+M_3\sin{\theta_2}+M_4\cos{\theta_2})^2\Big),\\
 \mathcal{B}_1=&(M_5-4k_1^5k_2^3\sin{(\theta_2-\theta_1)}+4k_1^3k_2^5\sin{(\theta_2-\theta_1)})^2\nonumber\\
 &+(M_6-4k_1^5k_2^3\cos{(\theta_2-\theta_1)}-4k_1^3k_2^5\cos{(\theta_2-\theta_1)})^2,\\
 \mathcal{A}_2=&4 (k_1^2 - k_2^2)^2  \Big((N_1 \cos{\theta_1} + N_2 \sin{\theta_1} + N_3 \cos{\theta_1 }
 + N_4 \sin{\theta_2})^2 \nonumber\\
 &+(-N_1 \sin{\theta_1} + N_2 \cos{\theta_1}  - N_3 \sin{\theta_2}+ N_4 \cos{\theta_2})^2\Big),\\
 \mathcal{B}_2=&(N_5 + 4 k_1^5 k_2^3 \sin{(\theta_1-\theta_2)} - 4 k_1^3 k_2^5  \sin{(\theta_1-\theta_2)})^2\nonumber\\
 &+(N_6 + 4 k_1^5 k_2^3 \cos{(\theta_1-\theta_2)} + 4 k_1^3 k_2^5 \cos{(\theta_1-\theta_2)})^2
\end{align}
\end{subequations}
with
\begin{align*}
  &M_1=2k_2^3(k_1^2-k_2^2)(k_1^4x-t),~~~M_2=-k_1^2k_2^3(3k_1^2+k_2^2),~~~
  M_3=2k_1^3(k_1^2-k_2^2)(k_2^4x-t),\\
 &M_4=k_1^3k_2^2(k_1^2+3k_2^2),~~~M_5=2 (k_1^2 - k_2^2)^2 (k_1^2 + k_2^2) (k_1^2k_2^2 x-t),~~~\\
 &M_6=k_1^6 k_2^2 + 6 k_1^4 k_2^4 + k_1^2 k_2^6 - 4 (k_1^2 - k_2^2)^2 t^2
 + 4 (k_1^2 - k_2^2)^2 (k_1^4 + k_2^4) t x - 4 k_1^4 k_2^4 (k_1^2 - k_2^2)^2 x^2,\\
 &N_1=2 k_1 (k_1^2 - k_2^2) ( k_2^4 x-t),~~~N_2=3 k_1^3 k_2^2 + k_1 k_2^4,~~~
 N_3=2 k_2 (k_1^2 - k_2^2) ( k_1^4 x-t ),\\
 &N_4=-k_1^4 k_2 - 3 k_1^2 k_2^3,~~~N_5=2 (k_1^2 - k_2^2)^2 (k_1^2 + k_2^2) (k_1^2 k_2^2 x-t ),\\
 &N_6=-k_1^6 k_2^2 - 6 k_1^4 k_2^4 - k_1^2 k_2^6 + 4 (k_1^2 - k_2^2)^2 t^2
 - 4 (k_1^2 - k_2^2)^2 (k_1^4 + k_2^4) t x + 4 k_1^4 k_2^4 (k_1^2 - k_2^2)^2 x^2,\\
 &\theta_1=-\frac{t + k_1^4 x}{k_1^2},~~~
 \theta_2=-\frac{t + k_2^4 x}{k_2^2}.
\end{align*}
To analyze the resulting algebraic 2SSs generated by $|q_1|^2$ and $|q_3|^2$,
we rewrite the solutions in coordinates $(X_1,t)$ and $(X_2,t)$, respectively, where
\begin{equation}
  X_1=x-\frac{t}{k_1^4},~~~~ X_2=x-\frac{t}{k_2^4}.
\end{equation}
For convenience, we denote the solitons with the two coordinates by $k_1$-soliton
and $k_2$-soliton correspondingly.
Then, fixing $X_1$ and $X_2$ respectively, and letting $|t|\to \infty$, we get the following results.
\begin{proposition}
  In the coordinate $(X_1,t)$, $|q_1|^2$ and $|q_3|^2$ asymptotically go to
  \begin{subequations}
    \begin{align}
   & |q_1|^2\sim \frac{4k_1^2}{1+ 4 k_1^4 X_1^2},~~~~|q_3|^2\sim\frac{4k_1^{-2}}{1 + 4 k_1^4 X_1^2},
   ~~~ (|t| \to \infty),
  \end{align}
  and in the coordinate $(X_2,t)$, $|q_1|^2$ and $|q_3|^2$ asymptotically go to
  \begin{align}
    |q_1|^2\sim \frac{4k_2^2}{1+ 4 k_2^4 X_2^2},~~~~|q_3|^2\sim\frac{4 k_2^{-2}}{1 + 4 k_2^4 X_2^2},
    ~~~ (|t| \to \infty).
  \end{align}
\end{subequations}
\end{proposition}
From this proposition, we can see that, asymptotically,  there are algebraic 1SSs
propagating along the line $x=k_1^{-4}t$ (vertex trajectory) with velocity $k_1^{-4}$ and
along the line $x=k_2^{-4}t$   with velocity $k_2^{-4}$, respectively.
For $k_1$-solitons, the amplitudes for $|q_1|^2$ and  $|q_3|^2$ are
\begin{equation*}
  \text{Amp}|_{(|q_1|^2)}=4 k_1^2,~~~~ \text{Amp}|_{(|q_3|^2)}=4 k_1^{-2},
\end{equation*}
and for $k_2$-solitons, the amplitudes for $|q_1|^2$ and  $|q_3|^2$ are
\begin{equation*}
  \text{Amp}|_{(|q_1|^2)}=4k_2^2,~~~~ \text{Amp}|_{(|q_3|^2)}=4 k_2^{-2}.
\end{equation*}
\captionsetup[figure]{labelfont={bf},name={Fig.},labelsep=period}
\begin{figure}[ht]
\centering
\subfigure[ ]{
\begin{minipage}[t]{0.45\linewidth}
\centering
\includegraphics[width=2.5in]{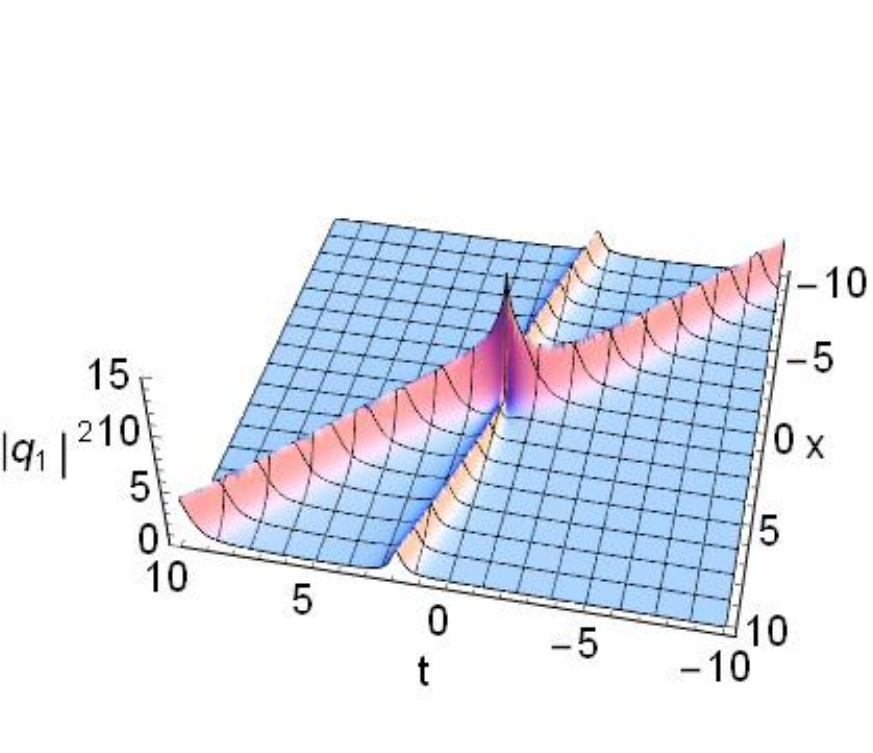}
\end{minipage}%
}%
\subfigure[ ]{
\begin{minipage}[t]{0.45\linewidth}
\centering
\includegraphics[width=2.6in]{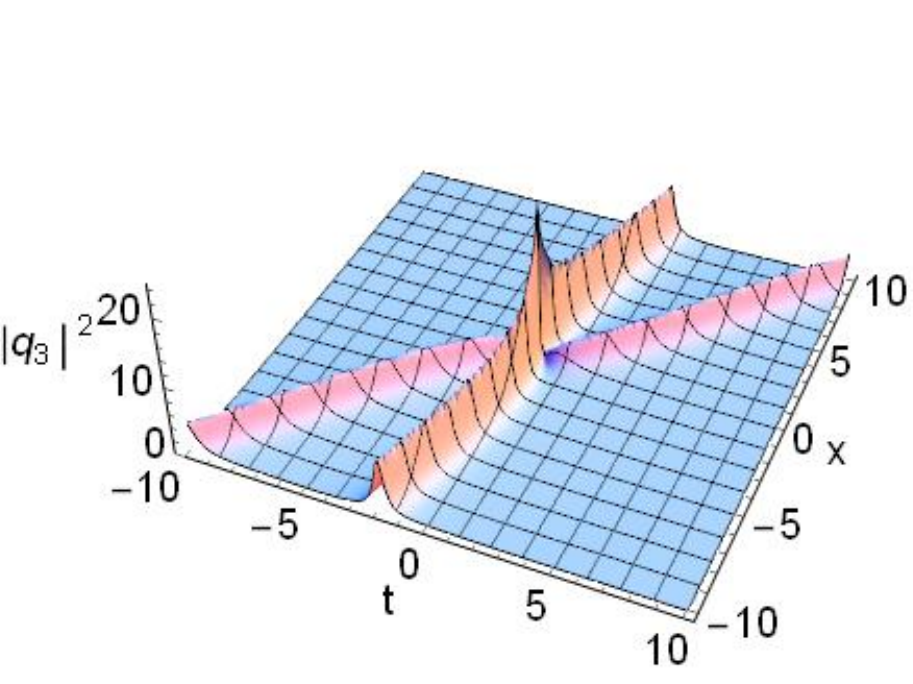}
\end{minipage}%
}
\subfigure[ ]{
\begin{minipage}[t]{0.45\linewidth}
\centering
\includegraphics[width=2.01in]{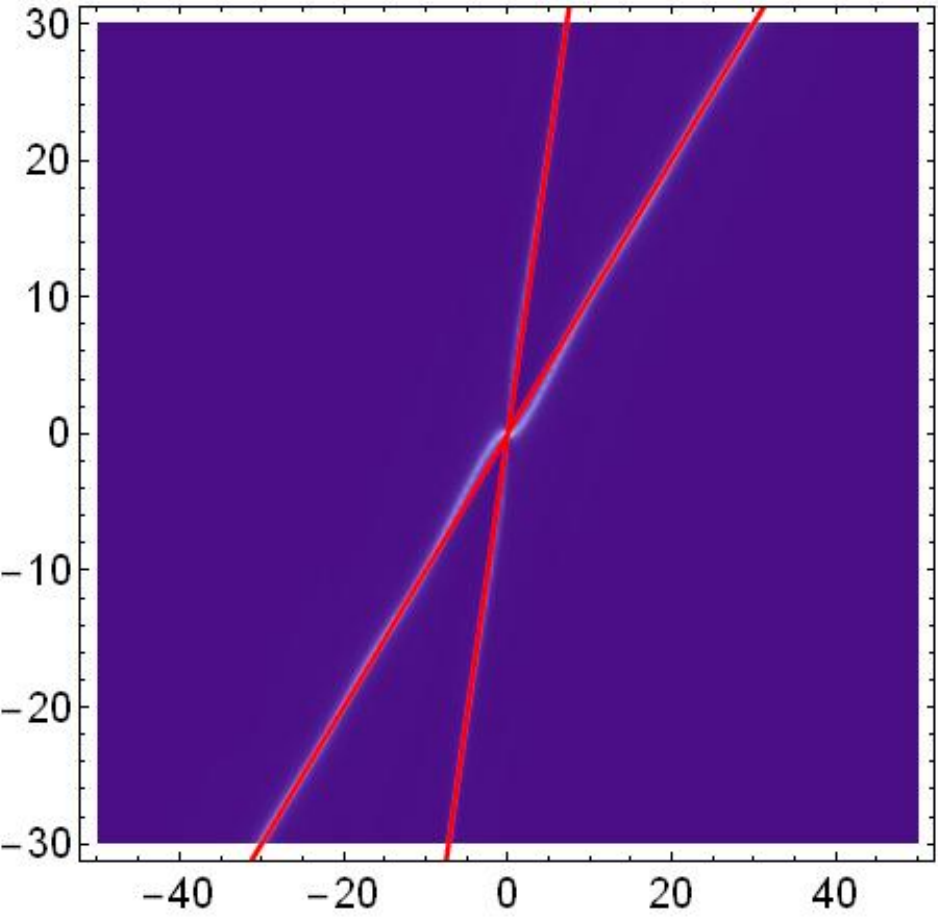}
\end{minipage}%
}
\subfigure[ ]{
\begin{minipage}[t]{0.45\linewidth}
\centering
\includegraphics[width=2.0in]{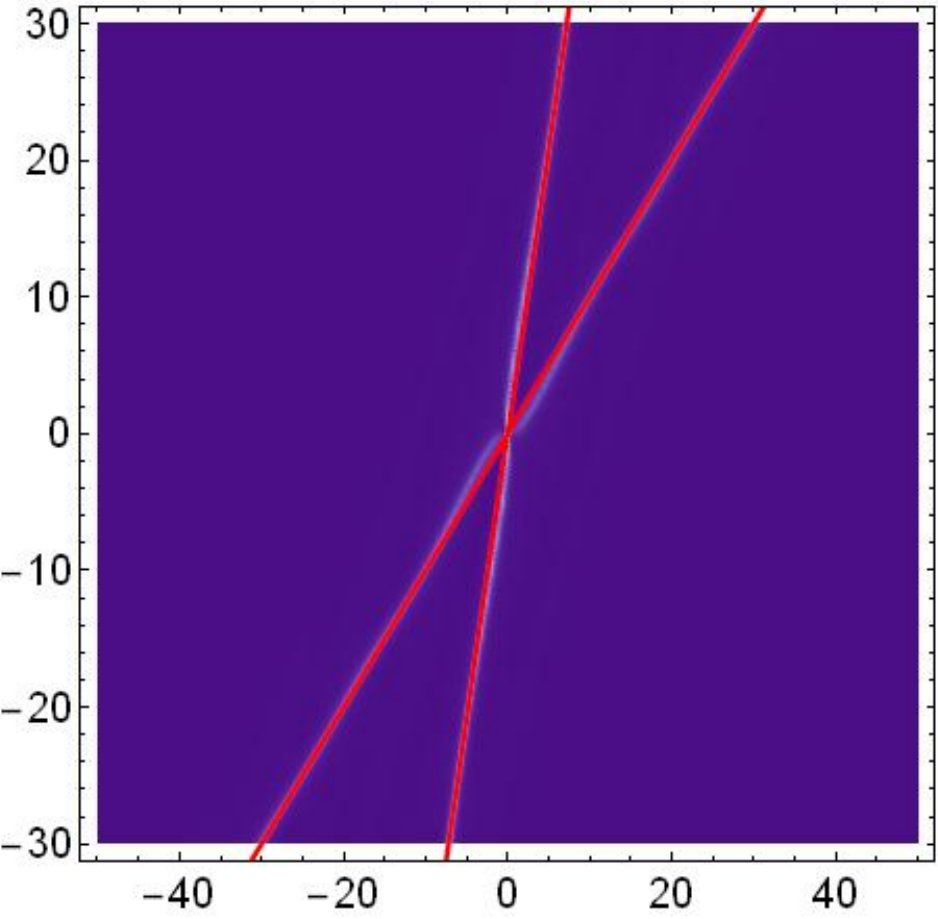}
\end{minipage}%
}
\caption{Shape and motion of algebraic 2SSs of the nonlocal  MTM   \eqref{non-cMTM}.
(a) Envelope $|q_1|^2$ given in \eqref{q1-alg-2ss} with $\delta=1,~k_1=1,~k_2=0.7$.
(b) Envelope of $|q_3|^2$ given in \eqref{q3-alg-2ss} with $\delta=1,~k_1=1,~k_2=0.7$.
(c) Density plot of (a) where the red lines are vertex trajectories $x=k_1^{-4}t$ and  $x=k_2^{-4}t$.
(d) Density plot of (b) where the red lines are vertex trajectories $x=k_1^{-4}t$ and  $x=k_2^{-4}t$.
 }
\label{F-7}
\end{figure}

As shown in Fig.\ref{F-7}, the derived algebraic solitons possess the following features:
near the interaction point there are apparent phase shifts for the two algebraic solitons, 
and these phase shifts gradually disappear when $|t|\gg 0$.
Such  features seem to common in the algebraic  type of soltions, see \cite{W-2021,WW-2022,LWZ-2022}.

\subsubsection{High order algebraic solitons}\label{sec-4-2-2}

In the following we consider
\begin{equation}\label{A-T-rational}
  A=\left(
   \begin{array}{ccccc}
    k &  0 & 0 & \cdots & 0 \\
     1 & k &0 & \cdots & 0 \\
     0 & 1 & k& \cdots & 0 \\
     \vdots & \ddots& \ddots & \ddots & \vdots \\
     0 & 0 & \cdots&1 & k
   \end{array}
 \right)_{ 2N\times 2N},~~
 T= \left(
    \begin{array}{ccccc}
     \sqrt{\delta}& 0 & 0 & \cdots & 0 \\
      0&  \sqrt{\delta} & 0 & \cdots & 0 \\
     0& 0&  \sqrt{\delta}& \cdots & 0 \\
      \vdots & \ddots & \ddots& \ddots & \vdots \\
     0& 0& \cdots& 0 &  \sqrt{\delta} \\
    \end{array}
  \right)_{ 2N\times 2N},
 \end{equation}
and
\begin{equation}
  S=AT.
\end{equation}
The corresponding vector $\phi$ and $\psi$ to compose solutions are
\begin{subequations}
\begin{align}\label{rational-A-T-phi-psi}
 & \phi =\Bigl(\mathrm{e}^{\eta(k)}, \partial_k(\mathrm{e}^{\eta(k)}),
  \frac{1}{2!}\partial_k^{2}( \mathrm{e}^{\eta(k)}),
  \cdots,  \frac{1}{(N-1)!}\partial_k^{N-1}(\mathrm{e}^{\eta(k)}) \Bigr)^{T},\\
  &\psi=S\phi(-x,-t),
  \end{align}
\end{subequations}
where $\eta(k )=\frac{i}{2}(k^{2}x+ k^{-2}t),~ k\in\mathbb{R},~\delta=\pm1$ and we have taken
         $C=(1,1,\cdots,1)^T$ for convenience.

The simplest case is of $N=1$, which yields the algebraic 1-soliton case \eqref{A-T-alg-1ss},
and corresponding envelops are the rational functions.

When $N=2$, we have
         \begin{equation}
          A=\left(
           \begin{array}{cccc}
            k &  0 & 0 & 0 \\
             1 & k &0 & 0 \\
             0 & 1 & k&  0 \\
             0 & 0 &1 & k
           \end{array}
         \right),~~
         T= \left(
          \begin{array}{cccc}
           \sqrt{\delta}& 0 & 0 &  0 \\
            0&  \sqrt{\delta} & 0 &  0 \\
           0& 0&  \sqrt{\delta}&  0 \\
           0& 0& 0 &  \sqrt{\delta} \\
          \end{array}
        \right),
         \end{equation}
         and
         \begin{subequations}
          \begin{align}\label{rational-A-T-phi-psi}
            & \phi =\Bigl(\mathrm{e}^{\eta(k)}, \partial_k(\mathrm{e}^{\eta(k)}),
             \frac{1}{2!}\partial_k^{2}( \mathrm{e}^{\eta(k)}),
            \frac{1}{(3)!}\partial_k^{3}(\mathrm{e}^{\eta(k)}) \Bigr)^{T},\\
             &\psi=S\phi(-x,-t)=AT\phi(-x,-t).
             \end{align}
            \end{subequations}
            \captionsetup[figure]{labelfont={bf},name={Fig.},labelsep=period}
\begin{figure}[ht]
\centering
\subfigure[ ]{
\begin{minipage}[t]{0.45\linewidth}
\centering
\includegraphics[width=2.5in]{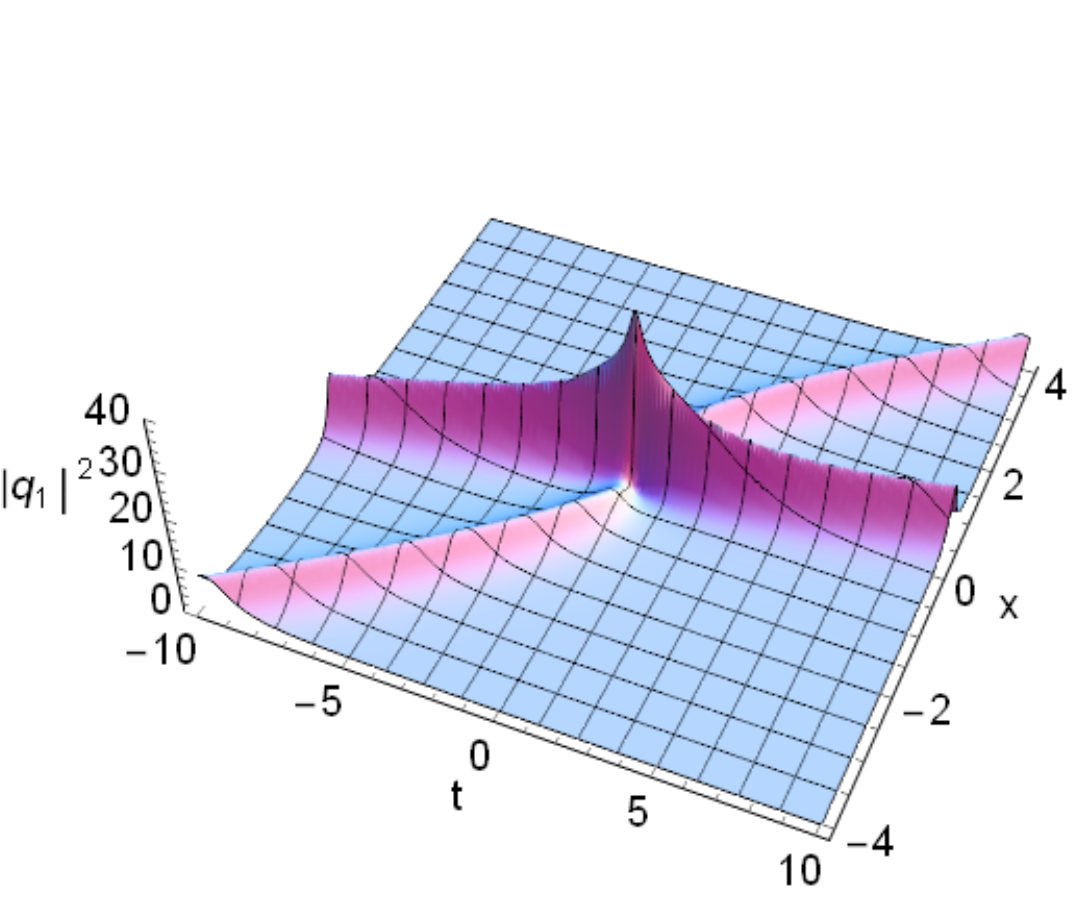}
\end{minipage}%
}%
\subfigure[ ]{
\begin{minipage}[t]{0.45\linewidth}
\centering
\includegraphics[width=2.5in]{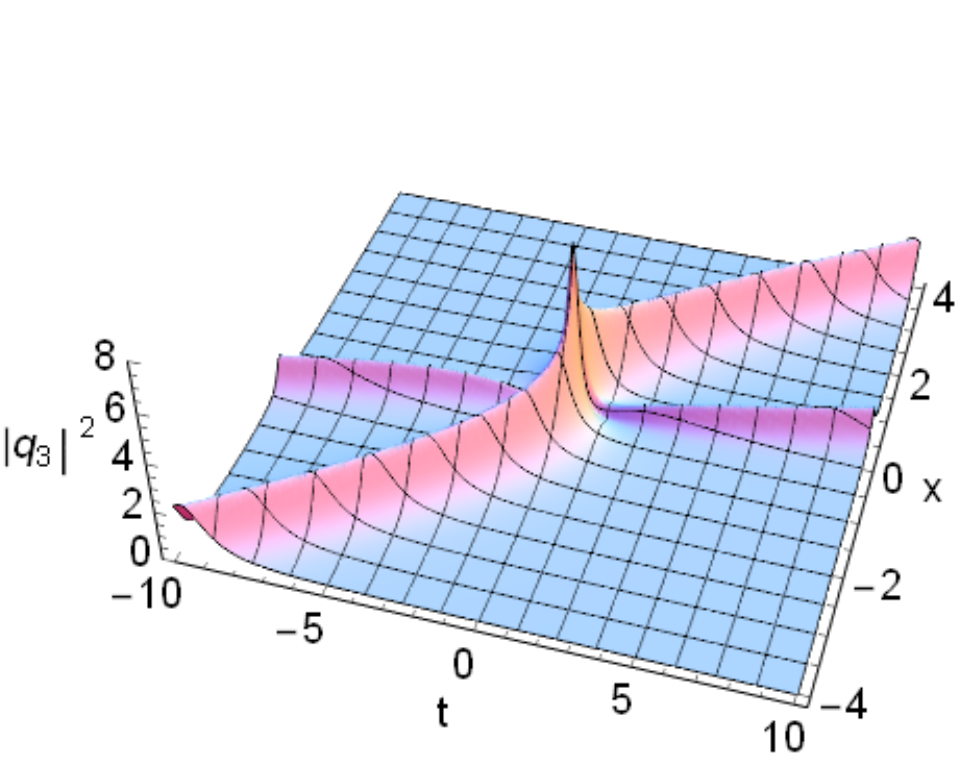}
\end{minipage}%
}
\subfigure[ ]{
\begin{minipage}[t]{0.45\linewidth}
\centering
\includegraphics[width=2.0in]{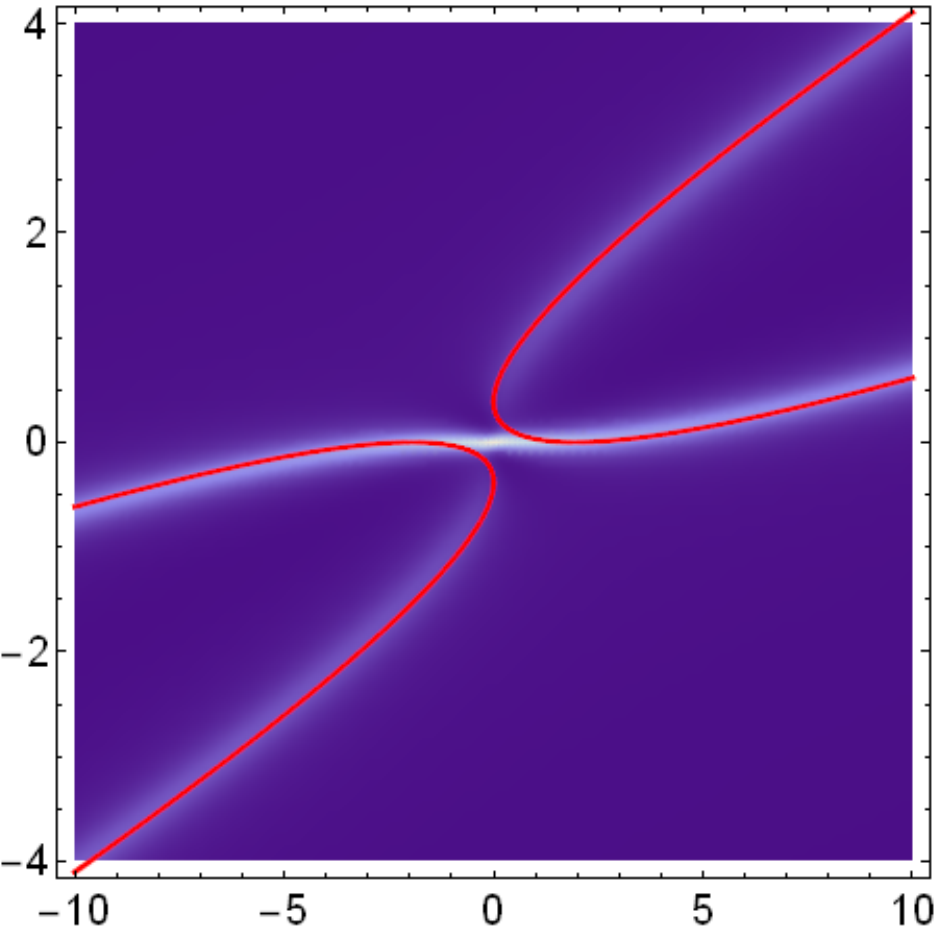}
\end{minipage}%
}
\subfigure[ ]{
\begin{minipage}[t]{0.45\linewidth}
\centering
\includegraphics[width=2.0in]{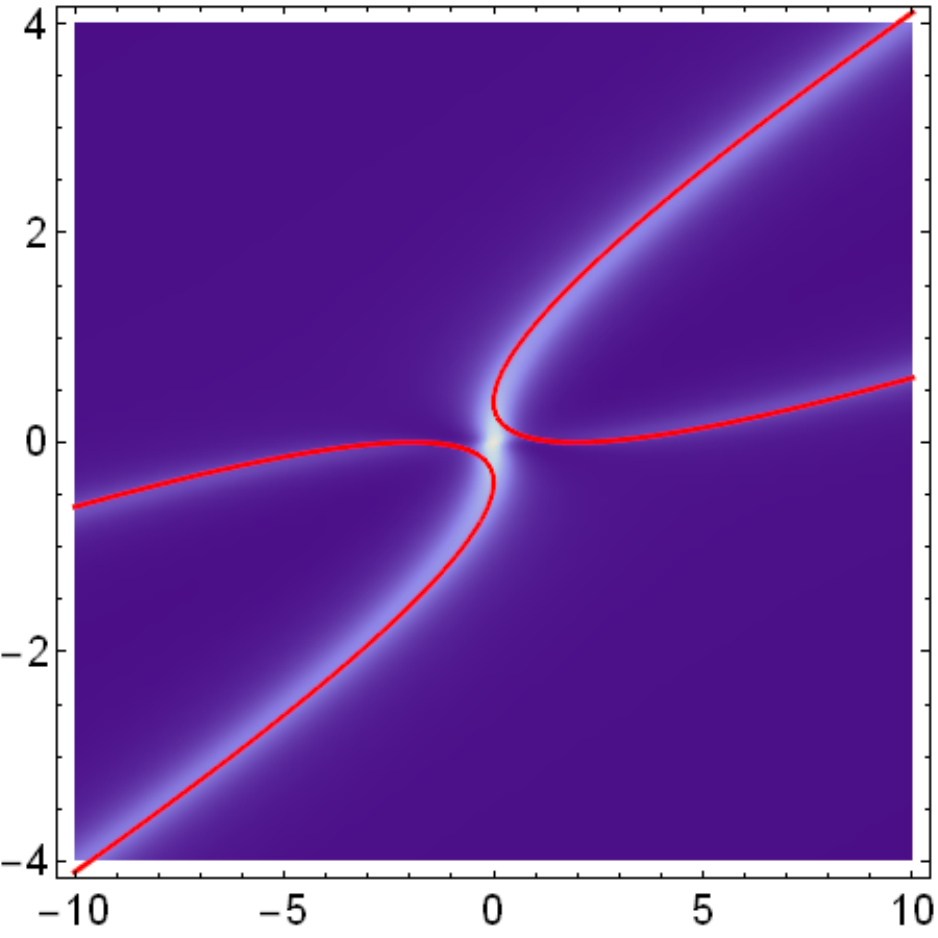}
\end{minipage}%
}
\caption{Shape and motion of the 2nd order algebraic soliton solutions of the nonlocal  MTM   \eqref{non-cMTM}.
(a) Envelope $|q_1|^2$ given in \eqref{q1-rational-2ss}with $\delta=1,~k=1.5$.
(b) Envelope of $|q_3|^2$ given in \eqref{q3-rational-2ss} with $\delta=1,~k=1.5$.
(c) Density plot of (a) where the four red curves are given in \eqref{curve1-2-3-4}.
(d) Density plot of (b) where the four red curves are given in \eqref{curve1-2-3-4}.
 }
\label{F-8}
\end{figure}
     The explicit formulae of $|q_1|^2$ and $|q_2|^2$ read
     \begin{subequations}\label{Abs-q1-q3-rational-2ss}
      \begin{align}
      |q_1|^2&=\frac{\mathcal{A'}_1}{\mathcal{B'}_1},\label{q1-rational-2ss}\\
     |q_3|^2&=\frac{\mathcal{A'}_2}{\mathcal{B'}_2},\label{q3-rational-2ss}
    \end{align}
  \end{subequations}
      where
      \begin{align*}
     \mathcal{A'}_1&=16 k^6 \left(64 t^6 + 64 k^{24} x^6 - 48 k^4 t^4 (-21 + 8 t x) - 48 k^{20} x^4 (-5 + 8 t x)\right.\\
     & \left.+60 k^8 t^2 (9 - 16 t x + 16 t^2 x^2) + 12 k^{16} x^2 (-3 + 48 t x + 80 t^2 x^2) \right.\\
     &\left.+ k^{12} (9 + 72 t x - 864 t^2 x^2 - 1280 t^3 x^3)\right),\\
     \mathcal{B'}_1&= \delta \left(256 t^8 + 256 k^{32} x^8 - 256 k^4 t^6 (-1 + 8 t x) - 256 k^{28} x^6 (-1 + 8 t x) \right.\\
     &\left.+ 32 k^8 t^4 (63 - 240 t x + 224 t^2 x^2) + 32 k^{24} x^4 (63 - 240 t x + 224 t^2 x^2)\right.\\
     &\left. - 16 k^{12} t^2 (-45 - 72 t x - 1776 t^2 x^2 + 896 t^3 x^3) -
   16 k^{20} x^2 (-45 - 72 t x - 1776 t^2 x^2 + 896 t^3 x^3) \right.\\
   &\left.+ k^{16} (9 - 288 t x + 30528 t^2 x^2 - 41984 t^3 x^3 + 17920 t^4 x^4)\right),\\
   \mathcal{A'}_2&=16 k^2 \left(64 t^6 + 64 k^{24} x^6 - 48 k^{20} x^4 (-21 + 8 t x) - 48 k^4 t^4 (-5 + 8 t x) \right.\\
   &\left.+ 60 k^{16} x^2 (9 - 16 t x + 16 t^2 x^2) + 12 k^8 t^2 (-3 + 48 t x + 80 t^2 x^2) \right.\\
   &\left.+ k^{12} (9 + 72 t x - 864 t^2 x^2 - 1280 t^3 x^3)\right),\\
     \mathcal{B'}_2&= \delta\left(256 t^8 + 256 k^{32} x^8 - 256 k^4 t^6 (-1 + 8 t x) \right.\\
    &\left. - 256 k^{28} x^6 (-1 + 8 t x) + 32 k^8 t^4 (63 - 240 t x + 224 t^2 x^2) + 32 k^{24} x^4 (63 - 240 t x + 224 t^2 x^2) \right.\\
    &\left.- 16 k^{12} t^2 (-45 - 72 t x - 1776 t^2 x^2 + 896 t^3 x^3) - 16 k^{20} x^2 (-45 - 72 t x - 1776 t^2 x^2 + 896 t^3 x^3)\right.\\
    &\left.+  k^{16} (9 - 288 t x + 30528 t^2 x^2 - 41984 t^3 x^3 + 17920 t^4 x^4)\right) .
  \end{align*}
 Obviously, they are   rational functions  with algebraic decaying, and we depict the waves in Fig.\ref{F-8}(a,b).
 More details about their dynamic behaves are presented in the following proposition.
 
 \begin{proposition}
 Asymptotically,   $|q_1|^2$ and $|q_3|^2$ given in \eqref{Abs-q1-q3-rational-2ss} 
  propagate along the following four curves (vertex trajectories)
(see the red curves in Fig.\ref{F-8}(c,d)):
\begin{subequations}\label{curve1-2-3-4}
  \begin{eqnarray}
&&x=\frac{t}{k^4} +\sqrt{-\frac{2\sqrt{3}}{k^6}t}-\frac{\sqrt{3}}{2k^2},~~t\leqslant 0,\label{curve-1}\\
&&x=\frac{t}{k^4} -\sqrt{-\frac{2\sqrt{3}}{k^6}t}-\frac{\sqrt{3}}{2k^2},~~t\leqslant 0,\label{curve-2}\\
&&x=\frac{t}{k^4} +\sqrt{\frac{2\sqrt{3}}{k^6}t}+\frac{\sqrt{3}}{2k^2},~~~~t\geqslant  0,\label{curve-3}\\
&&x=\frac{t}{k^4} -\sqrt{\frac{2\sqrt{3}}{k^6}t}+\frac{\sqrt{3}}{2k^2},~~~~t\geqslant 0.\label{curve-4}
\end{eqnarray}
\end{subequations}
Along the four curves, the amplitudes of the waves generate by $|q_1|^2$ and $|q_3|^2$ slowly change and approach to constants $4k^2$ and $4k^{-2}$ respectively when $|t|\gg0$.
More precisely, asymptotic properties are listed in Table.\ref{tab-1}.
Noted that the behavior with slowly varying amplitudes is a typical feature of such type of algebraic solitons \cite{W-2021,WW-2022,LWZ-2022}.
 \end{proposition}

 {\small{\begin{table}[h]
  \begin{center}
  \begin{tabular}{cccc}
    \hline
   \small branches
   &   \small asymptotic curve
   & \begin{minipage}{4.0 cm}{\small amplitude changing  \\\small with respect to $t$ for $|q_1|^2$  }\end{minipage}  & \begin{minipage}{4.0 cm}{ \small amplitude changing  \\\small with respect to $t$ for $|q_3|^2$ }\end{minipage}
   \\
   \hline
   \small left-up
  &    \small \eqref{curve-1}
     &  \small increase
     &  \small decrease \\
   \small  left-down
  &     \small \eqref{curve-2}
  &  \small decrease
  &  \small increase \\
  \small  right-up
  &    \small \eqref{curve-3}
     &  \small increase
     &  \small decrease \\
   \small  right-down
  &    \small \eqref{curve-4}
     &  \small decrease
     &  \small increase \\
  \hline
  \end{tabular}
  \caption{Asymptotic properties of $|q_1|^2$ and $|q_3|^2$ given by \eqref{Abs-q1-q3-rational-2ss} as depicted in Fig.\ref{F-8}}.
  \label{tab-1}
  \end{center}
  \end{table}
  }}
  \captionsetup[figure]{labelfont={bf},name={Fig.},labelsep=period}

\section{Conclusion}\label{sec-5}

In this paper we have studied integrability of the  nonlocal MTM \eqref{non-cMTM}
and its various types of solutions.
Bilinearization-reduction approach was employed.
We started from the 4-component system \eqref{31pKN-1},
which acts as the unreduced MTM and can be reduced to the classical MTM and nonlocal MTM.
Note that the Mikhailov model \eqref{pKN-1} is also related to the MTM.
In the paper we also explored the relations between  \eqref{31pKN-1} and \eqref{pKN-1}
and also their consistency in the nonlocal reduction.
In the bilinearization-reduction approach, the first step is to present bilinear form for 
the unreduced MTM \eqref{31pKN-1}  
and obtain solutions in double Wronskian form (see Theorem \ref{Theo-3-1} and its proof in Appendix \ref{app-A}).
Then, we imposed constraints on the Wronskian entry vectors $\phi$ and $\psi$
(see \eqref{ASn} in Theorem \ref{Theo-3-3}).
In principle, by solving \eqref{ASn} one can obtain all solutions of the nonlocal MTM
that are of zero background and formulated in double Wronskians. 
Explicit forms of  $\phi$ and $\psi$ were presented in Sec.\ref{sec-3-2},
where, in particular, in Case (2), we gave $\phi$ and $\psi$ that can generate
algebraic solitons and high order algebraic solitons.
Finally, dynamics of some obtained solutions were analyzed and illustrated.

Considering the classical MTM is an important equation in mathematical physics
while its nonlocal version is rarely studied compared with the classical case,
our research in this paper will bring more insight about the integrability and solutions of the nonlocal MTM.
Moreover, the 4-component unreduced massive Thirring system \eqref{31pKN-1} were first introduced in \cite{BG-1987}
where more reductions in different spaces of equations and solutions have been revealed. 
It should be possible to investigate these reductions of equations as well as solutions
from bilinear perspective. 
In addition, the bilinearization-reduction approach
works as well for the space/time-shifted integrable systems (e.g.\cite{AM-PLA-2021,LWZ-ROMP-2022}) and
and also for the solutions with nonzero backgrounds (e.g.\cite{ZLD-2023,DCCZ-2024}).
The study of these extensions for the nonlocal MTM will be considered elsewhere.

\subsection*{Acknowledgments}

This project is supported by the NSFC grant (Nos. 12271334 and 12201329).

\begin{appendix}

\section{Proof of Theorem \ref{Theo-3-1}}\label{app-A}

To prove the theorem, let us first recall  the following two lemmas that are useful in verifying  solutions
in the Wronskian technique.

\begin{lemma}\label{lemma 1}\cite{FreN-PLA-1983}
$$|\mathbf{M}, \mathbf{a},\mathbf{b}||\mathbf{M},\mathbf{c},\mathbf{d}|
-|\mathbf{M}, \mathbf{a},\mathbf{c}||\mathbf{M},\mathbf{b},\mathbf{d}|
+|\mathbf{M}, \mathbf{a},\mathbf{d}||\mathbf{M},\mathbf{b}, \mathbf{c}|=0,$$
where $\mathbf{M}$ is an arbitrary $N\times (N-2)$ matrix, and $\mathbf{a}$, $\mathbf{b}$,
$\mathbf{c}$ and $\mathbf{d}$ are $N$th-order column vectors.
\end{lemma}

\begin{lemma}\label{lemma 2}\cite{ZDJ-arxiv,ZhaZSZ-RMP-2014}
Let $\Xi=(a_{js})_{N\times N}$ be a $N\times N$ matrix with
column vectors $\{\alpha_j\}$. $\Gamma=(\gamma_{js})_{N\times N}$ is a  $N\times N$
operator matrix and each $\gamma_{js}$ is a operator. Then we have
\begin{equation*}
\sum^N_{j=1}|\Gamma_{j}\ast \Xi|=\sum^N_{j=1}|(\Gamma^{T})_{j}\ast \Xi^{T}|,
\end{equation*}
where
\begin{equation*}
|A_{j}\ast \Xi|=|\Xi_{1},\ldots,\Xi_{j-1},A\circ \Xi_{j},\Xi_{j+1},\ldots,\Xi_{N}|,
\end{equation*}
and
\begin{equation*}
A_{j}\circ B_{j}=(A_{1,j}B_{1,j},A_{2,j}B_{2,j},\ldots,A_{N,j}B_{N,j}),
\end{equation*}
in which $A_{j}=(A_{1,j},A_{2,j},\ldots,A_{N,j})^{T}$ and $B_{j}=(B_{1,j},B_{2,j},\ldots,B_{N,j})^{T}$
are $N$th-order column vectors.
\end{lemma}

Now we start the proof.
From  \eqref{wron-cond-x} one can calculate derivatives of $f$, $g$, $h$ and $s$:
\begin{align*}
&f_x= |\t{N-1}, N+1; \W{M-1}|+ |\t{N}; \W{M-2},M|,\nonumber \\
&f_t=-\frac{1}{4}(|0,\b{N}; \W{M-1}|+|\t{N}; -1,\t{M-1}|), \\
&g_x =\frac{1}{2(2i)^{N-1}}|A|(|\W{N-1},N+1; \t{M-1}|+|\W{N}; \t{M-2},M|),\\
&g_t =-\frac{1}{8(2i)^{N-1}}|A|(|-1,\t{N}; \t{M-1}|+|\W{N}; 0,\b{M-1}|),  \\
&h_x=-i \,(2i)^{N-1}|A|^{-1}(|\b{N-1}, N+1; \W{M}|+ |\b{N}; \W{M-1},M+1|),\\
&h_t=\frac{i}{4} \,(2i)^{N-1}|A|^{-1}(|1,\u{N}; \W{M}|+|\b{N}; -1,\t{M}|), \\
&s_x= |\t{N-1}, N+1; \t{M}|+ |\t{N}; \t{M-1},M+1|,\\
&s_t=-\frac{1}{4}(|0,\b{N}; \t{M}|+|\t{N}; 0,\b{M}|),\\
&\t h_t=\frac{1}{4(2i)^{N-1}}|A|(|-1,\t{N}; \W{M-2}|+|\W{N}; -1,\t{M-2}|),\\
&\b h_t=\frac{1}{4}(2i)^{N}|A|^{-1}(|1,\u{N}; \t{M+1}|+|\b{N}; 0,\b{M+1}|).
\end{align*}
Here $\u{N-1}$ denotes the consecutive columns $(\partial_{x}^{3}\phi,\cdots,\partial_{x}^{N-1}\phi)$.
Substituting them into equation \eqref{23a}, we have
\begin{align}
~& D_{t}\t h\cdot s+gf\nonumber\\
=\,& \t h_{t} s-\t h s_{t}+gf \nonumber\\
=\,& \frac{1}{2(2i)^{N-1}}|A|\Big[\frac{1}{2}|\t{N}; \t{M}|(|-1,\t{N}; \W{M-2}|+|\W{N}; -1,\t{M-2}|)
\nonumber\\
~&-\frac{1}{2}|\W{N}; \W{M-2}|(|0,\b{N}; \t{M}|+|\t{N}; 0,\b{M}|)+|\W{N}; \t{M-1}||\t{N}; \W{M-1}|\Big].
\label{eqd}
\end{align}
From the following identity
\begin{align*}
 |\t{N}; \t{M}| \Big((-2i \,\mathrm{Tr} (A^{-2}) |\W{N}; \W{M-2}|\Big)=
 |\W{N}; \W{M-2}| \Big((-2i \,\mathrm{Tr} (A^{-2}) |\t{N}; \t{M}|\Big),
\end{align*}
where $\mathrm{Tr}(A)$ stands for the trace of $A$, and making use of the above Lemma \ref{lemma 2},
we find the relation
\begin{align*}
|\t{N}; \t{M}| ( |-1, \t{N};  \W{M-2}|-|\W{N}; -1, \t{M-2}|)
= |\W{N}; \W{M-2}| ( |0, \b{N};  \t{M}|-|\t{N}; 0, \b{M}|).
\end{align*}
With it in hand, we reduce equation \eqref{eqd}   to
\begin{align}
~& D_{t}\t h\cdot s+gf\nonumber\\
 =&\, \frac{1}{2(2i)^{N-1}}|A|\Big[|\t{N}; \t{M}||\W{N}; -1,\t{M-2}|-|\W{N}; \W{M-2}||\t{N}; 0,\b{M}|
 +|\W{N}; \t{M-1}||\t{N}; \W{M-1}|\Big].
 \label{eqd1}
\end{align}
Next, using the dispersion relation of $\phi$ and $\psi$, i.e. \eqref{wron-cond-x},
we may rewrite some double Wronskians in the above equation as the following:
\begin{align*}
&|\W{N}; -1,\t{M-2}|
=(-1)^{M-1}(-2i)^{M+N}|A|^{-2}|\t{N+1}; 0, \overline{M-1}|,\\
&|\W{N}; \W{M-2}|
=(-1)^{M-1}(-2i)^{M+N}|A|^{-2}|\t{N+1}; \t{M-1}|,\\
&|\W{N}; \t{M-1}|
=(-1)^{M-1}(-2i)^{M+N}|A|^{-2}|\t{N+1}; \b{M}|.
\end{align*}
Substituting them into the right-hand side of  \eqref{eqd1},
and then making use of Lemma \ref{lemma 1},
we can immediately  find that the right-hand side of  \eqref{eqd1} vanishes,
which means bilinear  equation \eqref{23a} holds.

All other equations in \eqref{4MTM-d} can be proved along the same lines.
We skip details and complete the proof.

\section{Connection with the bilinear Mikhailov model}\label{app-B}

As pointed out in Proposition \ref{prop-2-1}, if
\begin{equation}\label{uv-0}
uv \to 0,~~~ \mathrm{as}~ x \to \pm \infty,
\end{equation}
the unreduced MTM \eqref{31pKN-1} can be transformed from the Mikhailov model \eqref{pKN-1}
through the transformation \eqref{trans 33}.
Note that the Mikhailov model \eqref{pKN-1}, via the transformation
\begin{equation}\label{tran}
u=\frac{g}{f},~~v=\frac{h}{s},
\end{equation}
can be written into bilinear form \cite{LWZ-2022}:
\begin{subequations}\label{4MTM-bb}
\begin{align}
& I_1=D_{x}D_{t}g\cdot f+gf=0 \label{22a},\\
& I_2=D_{x}D_{t}h\cdot s+hs=0 \label{22b},\\
& I_3=D_{x}D_{t}f\cdot s+iD_{x}g\cdot h=0, \label{22c} \\
& I_4=D_{t}f\cdot s+igh=0. \label{22d}
\end{align}
\end{subequations}
This implies that, under the assumption \eqref{uv-0},
the above equations should also serve as a bilinear form of the  unreduced MTM \eqref{31pKN-1}.
In the following let us explore more details of such connections.

First, bilinear equation \eqref{22d} together with \eqref{tran} indicates
\begin{equation}\label{uv-t}
uv=i\left(\ln \frac{f}{s}\right)_t.
\end{equation}
In light of the conserved relation \eqref{conser}, we may assume
\begin{equation}\label{qr}
qr=-i\left(\ln \frac{f}{s}\right)_x,
\end{equation}
and then from the definition of $\beta$, i.e. \eqref{beta},
we have
\begin{equation}\label{beta-a}
e^{i\beta}= \frac{f}{s}.
\end{equation}
It then follows from the transformation \eqref{trans 33} that
\begin{equation} \label{trans 5-2}
q_1=\left(\frac{g}{f}\right)_{x}\frac{f}{s}, ~~
~~q_2=\left(\frac{h}{s}\right)_{x}\frac{s}{f}, ~~
~~q_3=i\frac{g}{s}, ~~
~~q_4=-i\frac{h}{f}.
\end{equation}

Next, we demonstrate that \eqref{4MTM-bb} can serve as an alternative bilinear form of the
unreduced MTM \eqref{31pKN-1}, via the above transformation.
Substituting \eqref{trans 5-2} into \eqref{KN11} and \eqref{KN21} yields, respectively,
\begin{subequations}\label{eqs-1}
\begin{align}
& \eqref{KN11}=\frac{i}{fs} I_1+\!i\frac{D_xg\cdot f}{f^2s^2}I_4
-i\frac{g}{f^2s^2}\left(sD_xD_t f\cdot f + 2ihD_xg\cdot f\right)=0,\\
& \eqref{KN21}=\frac{i}{fs}I_2- i\frac{D_xh\cdot s}{f^2s^2}I_4-
i\frac{h}{f^2s^2}\left (fD_xD_t s\cdot s-2igD_xh\cdot s\right)=0,
\end{align}
\end{subequations}
where $I_1, I_2$ and $I_4$ are given as in \eqref{4MTM-bb}.
Meanwhile,  by using the following relations
\begin{subequations}
\begin{align}
& sD_xD_t f\cdot f=fD_xD_tf\cdot s-2f_xD_tf\cdot s+f(D_tf\cdot s)_x,\\
& fD_xD_t s\cdot s=sD_xD_tf\cdot s+2s_xD_tf\cdot s-s(D_tf\cdot s)_x,\\
& 2hD_xg\cdot f=fD_xg\cdot h+f(gh)_x-2f_xgh,\\
& 2gD_xh\cdot s=-sD_xg\cdot h+s(gh)_x-2s_xgh,
\end{align}
\end{subequations}
we can find that
\begin{subequations}\label{eqs-2}
\begin{align}
 sD_xD_t f\cdot f+2ihD_xg\cdot f &=f\, I_3+f (I_4)_x -2f_xI_4,\\
 fD_xD_t s\cdot s-2igD_xh\cdot s &=s\, I_3-s (I_4)_x +2s_xI_4.
\end{align}
\end{subequations}
Inserting them into \eqref{eqs-1}, we arrive at
\begin{subequations}\label{relation-NB}
\begin{align}
& \eqref{KN11}=\frac{i}{fs}I_1+i\frac{D_xg\cdot f}{f^2s^2}I_4
- i\frac{g}{f^2s^2}\Big(fI_3+fI_{4,x}-2f_xI_4\Big)=0,\\
& \eqref{KN21}=\frac{i}{fs}I_2-i\frac{D_xh\cdot s}{f^2s^2}I_4
-i\frac{h}{f^2s^2}\Big(sI_3-sI_{4,x}+2s_xI_4\Big)=0,
\end{align}
\end{subequations}
which means \eqref{4MTM-bb} provides a set of bilinear form for the first two equations \eqref{KN11} and \eqref{KN21}.

For equation \eqref{KN31}, with the transformation \eqref{trans 5-2}, it reads
\begin{align}
\eqref{KN31}=
-\Big(\frac{g}{f}\cdot \frac{f}{s}\Big)_x+\Big(\frac{g}{f}\Big)_x\frac{f}{s}
+i\Big(\frac{g}{f}\Big)_x\Big(\frac{h}{s}\Big)_x\frac{g}{s}=0,
\end{align}
i.e.
\begin{equation}\label{eqs-4}
  \Big(\frac{g}{f}\Big)_x \Big(\frac{h}{s}\Big)_x=-i \Big(\ln{\frac{f}{s}}\Big)_x .
\end{equation}
Recalling the conservation law \eqref{con-q} together with \eqref{trans 5-2} which  indicates
\begin{equation}\label{eqs-4}
 \partial_t\left[ \Big(\frac{g}{f}\Big)_x \Big(\frac{h}{s}\Big)_x\right]
 =-  \left( \frac{gh}{fs}\right)_x,
\end{equation}
and recalling bilinear equation \eqref{22d} which indicates
\begin{equation}\label{eqs-4}
\partial_t \left[ \Big(\ln{\frac{f}{s}}\Big)_x \right]
=-i    \left( \frac{gh}{fs}\right)_x,
\end{equation}
if we assume both $\Big(\frac{g}{f}\Big)_x \Big(\frac{h}{s}\Big)_x$
and $\Big(\ln{\frac{f}{s}}\Big)_x $ tend to zero when $|t|\to \infty$,
we can conclude that \eqref{eqs-4} holds, so does equation \eqref{KN31}.
One can have a same result for equation \eqref{KN41}.
Moreover, comparing the two transformations \eqref{trans 5} and \eqref{trans 5-2},
we find
\begin{equation}
\t h=\frac{D_xg\cdot f}{s},~~~~~
\b h=\frac{D_xh\cdot s}{f},
\end{equation}
which are nothing but the bilinear equations \eqref{23d} and \eqref{23e}.

In conclusion, \eqref{4MTM-bb} does provides a set of bilinear formulation for
the unreduced MTM \eqref{31pKN-1} under the aforementioned assumption.
However, as we have pointed out in Remark \ref{Rem-2.1},
the transformation \eqref{trans 5-2} and bilinear form \eqref{4MTM-bb}
rely on the conservation laws \eqref{con-q}, \eqref{conser}
and the significance of $\beta=\partial^{-1}_x qr$,
both of which require that $q,r,q_i\to 0$ as $x\to \pm \infty$.
In principle, the bilinear form \eqref{4MTM-bb} is not applicable to
those solutions whose envelopes do not decrease
(e.g.  Fig.\ref{F-2}(a), Fig.\ref{F-3}, Fig.\ref{F-4}(c,d) and Fig.\ref{F-5}(b) in this paper).

\end{appendix}

\end{document}